%% file: Antenna_Selection_AirComp.tex
\documentclass[12pt, draftclsnofoot, onecolumn]{IEEEtran}
\usepackage{lineno}
\usepackage{cite}
\usepackage{amsmath,amssymb,amsfonts}
\usepackage{graphicx}
\usepackage{xcolor}
\usepackage{balance}
\usepackage{cite}
\usepackage{hyperref}
\usepackage{amsthm}
\usepackage{float}
\usepackage{mathrsfs}
\usepackage{mathtools}
\usepackage{bm}
\usepackage{paralist}
\usepackage{array}
\usepackage{url}
\usepackage{caption}
\usepackage{setspace}
\usepackage{dsfont}
\usepackage{bbm}
\usepackage{algorithm}
\usepackage{algorithmic}
\usepackage[nolist]{acronym}
\usepackage{tikz,pgfplots}
\usepackage{caption}
\usepackage{subcaption}

\DeclareMathOperator*{\argmin}{argmin}
\newtheorem{remark}{Remark}
\newtheorem{theorem}{Theorem}

\newtheorem{lemma}{Lemma}

\newtheorem{corollary}{Corollary}

\makeatletter
\newcommand{\biggg}{\bBigg@{3}}
\newcommand{\Biggg}{\bBigg@{3.5}}
\makeatother
\def\BibTeX{{\rm B\kern-.05em{\sc i\kern-.025em b}\kern-.08em
    T\kern-.1667em\lower.7ex\hbox{E}\kern-.125emX}}

\begin{document}
	\newcommand{\her}{\mathsf{H}}
	\newcommand{\setX}{\mathbbmss{X}}
	\newcommand{\setL}{\mathbbmss{L}}
	\newcommand{\setA}{\mathbbmss{A}}
	\newcommand{\setP}{\mathbbmss{P}}
	\newcommand{\setR}{\mathbbmss{R}}
	\newcommand{\setV}{\mathbbmss{V}}
	\newcommand{\setS}{\mathbbmss{S}}
	\newcommand{\setW}{\mathbbmss{W}}
	\newcommand{\setJ}{\mathbbmss{J}}
	\newcommand{\setZ}{\mathbbmss{Z}}
	\newcommand{\setB}{\mathbbmss{B}}
	\newcommand{\setO}{\mathbbmss{O}}
	\newcommand{\setC}{\mathbbmss{C}}
	\newcommand{\setU}{\mathbbmss{U}}
	\newcommand{\setD}{\mathbbmss{D}}
	\newcommand{\setF}{\mathbbmss{F}}
	\newcommand{\setH}{\mathbbmss{H}}
	
	\newcommand{\Ex}[2]{ \mathbbmss{E}_{#2} \left\lbrace #1 \right\rbrace }
	\newcommand{\Ind}[1]{ \mathds{I} \left\lbrace #1 \right\rbrace }
	
	\newcommand{\minSub}[4]{
		&\min_{#2} #1 \tag{$#4$} \\
		&{\rm{s.t.}}\quad
		\begin{array}{l}
			#3
		\end{array}
		\nonumber
	}

	\newcommand{\rmP}{P_{\max}}
	\newcommand{\ave}{\mathrm{P}}
	\newcommand{\rmj}{\mathrm{j}}
	\newcommand{\rmc}{\mathrm{c}}
	\newcommand{\rmp}{\mathrm{p}}
	\newcommand{\rmq}{\mathrm{q}}
	\newcommand{\rmF}{\mathrm{F}}
	\newcommand{\rmH}{\mathrm{H}}
	\newcommand{\rmR}{\mathrm{R}}
	\newcommand{\rmg}{\mathrm{g}}
	\newcommand{\rmG}{\mathrm{G}}
	\newcommand{\rmQ}{\mathrm{Q}}
	\newcommand{\rmT}{\mathrm{T}}
	\newcommand{\rmA}{\mathrm{A}}
	\newcommand{\rmK}{\mathrm{K}}
	\newcommand{\rmI}{\mathrm{I}}
	\newcommand{\rmS}{\mathrm{S}}
	\newcommand{\lams}{\lambda^{\mathsf{s}}}
	
	\newcommand{\out}{\mathrm{out}}
	\newcommand{\In}{\mathrm{in}}
	
	\newcommand{\rsim}{\mathsf{Sim}^\mathsf{R}}
	\newcommand{\vecc}{\mathrm{vec}}
	\newcommand{\subto}{\mathop{\mathrm{s.t.}\;}}

	\newcommand{\sfF}{\mathsf{F}}
	\newcommand{\sfI}{\mathsf{I}}
	\newcommand{\sfT}{\mathsf{T}}
	\newcommand{\sfC}{\mathsf{C}}
	\newcommand{\sfM}{\mathsf{M}}
	\newcommand{\sfQ}{\mathsf{Q}}
	\newcommand{\sfE}{\mathsf{E}}
	\newcommand{\sfD}{\mathsf{D}}
	\newcommand{\sfW}{\mathsf{W}}
	\newcommand{\sfr}{\mathsf{r}}
	\newcommand{\sfR}{\mathsf{R}}
	\newcommand{\sfd}{\mathsf{d}}
	\newcommand{\sfZ}{\mathsf{Z}}
	\newcommand{\sfL}{\mathsf{L}}
	\newcommand{\sfK}{\mathsf{K}}
	\newcommand{\sfp}{\mathsf{p}}
	\newcommand{\sfa}{\mathsf{a}}

	\newcommand{\mah}{\mathcal{H}}
	\newcommand{\maf}{\mathcal{F}}
	\newcommand{\maT}{\mathcal{T}}
	\newcommand{\maP}{\mathcal{P}}
	\newcommand{\mae}{\mathcal{E}}
	\newcommand{\maz}{\mathcal{Z}}
	\newcommand{\mas}{\mathcal{S}}
	\newcommand{\mav}{\mathcal{V}}
	\newcommand{\mab}{\mathcal{B}}
	\newcommand{\man}{\mathcal{N}}
	\newcommand{\mam}{\mathcal{M}}
	\newcommand{\maD}{\mathcal{D}}
	\newcommand{\mao}{\mathcal{O}}
	\newcommand{\maq}{\mathcal{Q}}
	\newcommand{\mg}{\mathcal{G}}
	\newcommand{\mai}{\mathcal{I}}
	\newcommand{\maj}{\mathcal{J}}
	\newcommand{\mak}{\mathcal{K}}
	\newcommand{\mac}{\mathcal{C}}
	\newcommand{\maw}{\mathcal{W}}
	\newcommand{\mal}{\mathcal{L}}

	\newcommand{\bhxx}{\mathbf{x}}
	\newcommand{\bo}{\mathbf{o}}
	\newcommand{\bxx}{\mathbf{x}}
	\newcommand{\mh}{\mathbf{h}}
	\newcommand{\bss}{\mathbf{s}}
	\newcommand{\bww}{\mathbf{w}}
	\newcommand{\buu}{\mathbf{u}}
	\newcommand{\bvv}{\mathbf{v}}
	\newcommand{\byy}{\mathbf{y}}
	\newcommand{\bgg}{\mathbf{g}}
	\newcommand{\bff}{\mathbf{f}}
	\newcommand{\bhzz}{{\mathbf{z}}}
	\newcommand{\bmm}{{\mathbf{m}}}
	
	\newcommand{\Th}{{\mathrm{Th}}}
	
	\newcommand{\bd}{\mathbf{d}}
	
	\newcommand{\rmx}{\mathrm{x}}
	\newcommand{\rms}{\mathrm{s}}
	\newcommand{\rmw}{\mathrm{w}}
	\newcommand{\rmu}{\mathrm{u}}
	\newcommand{\rmh}{\mathrm{h}}
	\newcommand{\rmv}{\mathrm{v}}
	\newcommand{\rmy}{\mathrm{y}}
	\newcommand{\papr}{\mathrm{PAPR}}

	\newcommand{\bh}{{\mathbf{h}}}
	\newcommand{\btheta}{{\boldsymbol{\theta}}}
	\newcommand{\bzeta}{{\boldsymbol{\zeta}}}
	\newcommand{\blambda}{{\boldsymbol{\lambda}}}
	\newcommand{\bxi}{{\boldsymbol{\xi}}}
	\newcommand{\bmu}{\boldsymbol{\mu}}
	\newcommand{\bx}{{\boldsymbol{x}}}
	\newcommand{\hx}{{\hat{x}}}
	\newcommand{\tx}{{\tilde{x}}}
	\newcommand{\tv}{{\tilde{v}}}
	\newcommand{\hv}{{\hat{v}}}
	\newcommand{\htt}{{\hat{t}}}
	\newcommand{\vv}{\mathrm{v}}
	\newcommand{\hxx}{\hat{\mathrm{x}}}
	\newcommand{\xx}{\mathrm{x}}
	\newcommand{\yy}{\mathrm{y}}
	\newcommand{\btx}{\boldsymbol{\tilde{x}}}
	\newcommand{\btv}{\boldsymbol{\tilde{v}}}
	\newcommand{\bhx}{{\boldsymbol{\hat{x}}}}
	\newcommand{\set}[1]{\left\lbrace#1\right\rbrace}
	\newcommand{\Diag}[1]{\mathrm{Diag}\left\lbrace #1 \right\rbrace}
	\newcommand{\brc}[1]{\left( #1 \right) }
	\newcommand{\abs}[1]{\left\vert #1 \right\vert }
	\newcommand{\norm}[1]{\left\Vert #1 \right\Vert }
	\newcommand{\inner}[1]{\left\langle #1 \right\rangle }
	\newcommand{\dbb}[1]{\Pi_{\rm d}\left( #1 \right) }
	\newcommand{\dbc}[1]{\left[ #1 \right] }
	\newcommand{\itr}[1]{^{\left( #1 \right)} }
	\newcommand{\alg}{\mathrm{Alg}}
	\newcommand{\mar}{\mathcal{R}}

	\newcommand{\ty}{{\tilde{y}}}
	\newcommand{\rme}{\mathrm{e}}
	\newcommand{\tz}{{\tilde{z}}}
	\newcommand{\mfPt}{\tilde{\mathfrak{P}}}
	\newcommand{\mfIt}{\tilde{\mathfrak{I}}}
	\newcommand{\makt}{\tilde{\mathcal{K}}}
	
	\newcommand{\bz}{{\boldsymbol{z}}}
	\newcommand{\ba}{{\boldsymbol{a}}}
	\newcommand{\bg}{{\boldsymbol{g}}}
	\newcommand{\bs}{{\boldsymbol{s}}}
	\newcommand{\brr}{{\mathbf{c}}}
	\newcommand{\br}{{\mathbf{r}}}
	\newcommand{\bcc}{{\mathbf{C}}}
	\newcommand{\bp}{{\mathbf{p}}}
	\newcommand{\bP}{{\mathbf{P}}}
	\newcommand{\bu}{{\boldsymbol{u}}}
	\newcommand{\bv}{{\boldsymbol{v}}}
	\newcommand{\bc}{{\boldsymbol{c}}}
	\newcommand{\bw}{{\boldsymbol{w}}}
	
	\newcommand{\pmf}{{\mathrm{P}}}
	\newcommand{\pdf}{{\mathrm{f}}}
	\newcommand{\cdf}{{\mathrm{F}}}

	\newcommand{\bzz}{{\mathbf{z}}}
	\newcommand{\baa}{{\mathbf{a}}}

	\newcommand{\mfP}{\mathfrak{P}}
	\newcommand{\mfI}{\mathfrak{I}}
	\newcommand{\dif}{\mathrm{d}}
	
	\newcommand{\by}{{\boldsymbol{y}}}
	\newcommand{\hy}{{\hat{y}}}
	\newcommand{\bhy}{{\boldsymbol{\hat{y}}}}
	
	\newcommand{\bn}{{\mathbf{n}}}
	\newcommand{\hn}{{\hat{n}}}
	\newcommand{\be}{{\mathbf{e}}}
	
	\newcommand{\trp}{\mathsf{T}}
	
	\newcommand{\mA}{\mathbf{A}}
	\newcommand{\mF}{\mathbf{F}}
	\newcommand{\mR}{\mathbf{R}}
	\newcommand{\ma}{\mathbf{A}}
	\newcommand{\mI}{\mathbf{I}}
	\newcommand{\mone}{\mathbf{1}}
	\newcommand{\mJ}{\mathbf{J}}
	\newcommand{\mG}{\mathbf{G}}
	\newcommand{\mQ}{\mathbf{Q}}
	\newcommand{\mS}{\mathbf{S}}
	\newcommand{\mW}{\mathbf{W}}
	\newcommand{\mU}{\mathbf{U}}
	\newcommand{\mGam}{\mathbf{A}}
	\newcommand{\mD}{\mathbf{D}}
	\newcommand{\mM}{\mathbf{M}}
	\newcommand{\mLA}{\mathbf{\Lambda}}
	\newcommand{\mTheta}{\mathbf{\Theta}}
	\newcommand{\mSigma}{\mathbf{\Sigma}}
	\newcommand{\mPsi}{\mathbf{\Psi}}
	\newcommand{\mXi}{\mathbf{\Xi}}
	\newcommand{\mY}{\mathbf{Y}}
	\newcommand{\mX}{\mathbf{X}}
	\newcommand{\mV}{\mathbf{V}}
	\newcommand{\mtV}{\tilde{\mathbf{V}}}
	\newcommand{\mB}{\mathbf{B}}
	\newcommand{\mT}{\mathbf{T}}
	\newcommand{\mH}{\mathbf{H}}
	\newcommand{\mL}{\mathbf{L}}
	\newcommand{\mC}{\mathbf{C}}
	\newcommand{\rmE}{\mathrm{Erg}}
	\newcommand{\bomega}{\boldsymbol{\omega}}

	\newcommand{\md}{\mathrm{D}}
	\newcommand{\E}{\mathbb{E}}
	\newcommand{\MSE}{\mathsf{MSE}}
	\newcommand{\mse}{\mathsf{MSE}}
	\newcommand{\hp}{\hat{p}}
	\newcommand{\rs}{{\mathsf{rs}}}
	\newcommand{\rsb}{{\mathsf{rsb}}}
	\newcommand{\rf}{{\mathrm{RF}}}
	\newcommand{\rss}{{\mathrm{RSS}}}
	\newcommand{\rls}{{\mathrm{rls}}}
	\newcommand{\rzf}{{\mathrm{rzf}}}
	\newcommand{\C}{\mathbb{C}}	
\begin{acronym}
	\acro{mimo}[MIMO]{multiple-input multiple-output}
	\acro{rf}[RF]{radio frequency}
	\acro{as}[AS]{antenna selection}
	\acro{mse}[MSE]{mean square error}
	\acro{irs}[IRS]{intelligent reflecting surface}
	\acro{sinr}[SINR]{signal to noise and interference ratio}
	\acro{snr}[SNR]{signal to noise ratio}
	\acro{ps}[PS]{parameter server}
	\acro{csi}[CSI]{channel state information}
	\acro{iid}[i.i.d]{independent and identically distributed}
	\acro{dc}[DC]{difference-of-convex-functions}
	\acro{dcp}[DCP]{DC programming}
	\acro{qcqp}[QCQP]{quadratically constrained quadratic programming} 
	\acro{sdp}[SDP]{semi-definite program} 
	\acro{fl}[FL]{federated learning} 
	\acro{otaFL}[OTA-FL]{over-the-air federated learning} 
	\acro{aircomp}[AirComp]{over-the-air computation} 
	\acro{awgn}[AWGN]{additive white Gaussian noise} 
	\acro{mac}[MAC]{multiple access channel} 
	\acro{mp}[MP]{matching pursuit}
	\acro{np}[NP]{nondeterministic polynomial time}
	\acro{lmi}[LMI]{linear matrix inequality}
	\acro{ao}[AO]{alternating optimization}
	\acro{fw}[FW]{Frank-Wolfe}
	\acro{amp}[AMP]{approximate message passing}
	\acro{tdd}[TDD]{time-devision-duplex}
	\acro{pdd}[PDD]{penalty dual decomposition}
	\acro{al}[AL]{augmented Lagrangian}
	\acro{ota-fl}[OTA-FL]{over-the-air federated learning }
	\acro{lars}[LARS]{least angle regression}
	\acro{ist}[IST]{iterative soft thresholding}
	\acro{dl}[DL]{deep learning}
	\acro{zf}[ZF]{zero forcing}
	\acro{fista}[FISTA]{fast iterative soft thresholding algorithm}
	\acro{mmse}[MMSE]{minimum mean squared error }
	\acro{adc}[ADC]{analog-to-digital converter}
	\acro{sgd}[SGD]{stochastic gradient descent }
	\acro{cnn}[CNN]{convolutional neural network}
	\acro{dnn}[DNN]{deep neural network}
	\acro{nn}[NN]{neural network}
	\acro{mmw}[mmW]{millimeter wave}
	\acro{AoA}[AoA]{angle of arrival}
 \acro{Lasso}[Lasso]{least absolute shrinkage and selection operator}
 
\end{acronym}

\title{Joint Antenna Selection and Beamforming for Massive MIMO-enabled Over-the-Air Federated Learning \\
}

\author{
	\IEEEauthorblockN{
		Saba Asaad, \textit{Member IEEE}, 
		Hina Tabassum, \textit{Senior Member},
	   Chongjun Ouyang, \textit{Member IEEE},
	   Ping Wang, \textit{Fellow, IEEE}.
\thanks{Saba Asaad, Hina Tabassum and Ping Wang are with the Department of Electrical Engineering and Computer Science at York University, Toronto, Canada; emails: \texttt{asaads@yorku.ca, hina.tabassum@lassonde.yorku.ca, ping.wang@lassonde.yorku.ca}. Chongjun Ouyang is with the School of Information and Communication Engineering, Beijing University of Posts and Telecommunications, Beijing, China; email: \texttt{dragonaim@bupt.edu.cn}.}
}
}
\IEEEoverridecommandlockouts
\maketitle
\vspace{-1.8cm}
\begin{abstract}
	Over-the-air federated learning (OTA-FL) is an emerging  technique to reduce the computation and communication overload at the \ac{ps} caused by the orthogonal transmissions of the model updates in conventional federated learning (FL).
 This reduction is achieved at the expense of introducing aggregation error that can be efficiently suppressed by means of receive beamforming via large array-antennas. 
 This paper studies OTA-FL in massive multiple-input multiple-output (MIMO) systems by considering a realistic scenario in which the edge server, despite its large antenna array, is restricted in the number of radio frequency (RF)-chains.  For this setting, the beamforming for over-the-air  model aggregation needs to be addressed jointly with antenna selection. This leads to an NP-hard problem due to the combinatorial nature of the optimization. We tackle this problem via two different approaches. In the first approach, we use the penalty dual decomposition (PDD) technique to develop a two-tier algorithm for joint antenna selection and beamforming. The second approach interprets the antenna selection task as a sparse recovery problem and develops two iterative joint algorithms based on the \ac{Lasso} and fast iterative soft-thresholding methods. Convergence and complexity analysis is presented for all the schemes. The numerical investigations depict that the algorithms based on the sparse recovery techniques outperform the PDD-based algorithm, when the number of RF-chains at the edge server is much smaller than its array size. However, as the number of RF-chains increases, the PDD approach starts to be superior. Our simulations further depict that learning performance  with all the antennas being active at the \ac{ps} can be closely tracked by selecting less than $20\%$ of the antennas at the \ac{ps}.
\end{abstract}
\begin{IEEEkeywords}
Over-the-air federated learning, over-the-air computation, distributed machine learning, antenna selection, radio-frequency chain, beamforming, massive MIMO.
\end{IEEEkeywords}

\section{Introduction}

Recently, \ac{fl} has been hailed as a key distributed machine learning technique, capable of training a global model collaboratively by exchanging local model updates across clients\cite{konevcny2016federated}. 
The core idea of \ac{fl} is simple: it suggests the distributed devices in the network to keep private and sensitive information on their local storage and share only the locally-trained models periodically with the parameter server (\ac{ps}). The \ac{ps} aggregates these local models into a global model and broadcasts it to the clients for a new local training round. This process continues until the global model parameters converge \cite{bonawitz2019towards}. 

 
%
%

\ac{fl} was originally proposed by Google to support distributed learning over \textit{wire-line} connected systems \cite{kairouz2021advances}. The communication network was hence modeled by ideal links in the earlier research works, and the communication limits were taken into account as simple restricted budget constraints in the network \cite{jordan2018communication}. Nevertheless, recent advances of intelligent wireless edge-devices and  wireless technologies enable \ac{fl}  in wireless networks \cite{chen2020joint}. Unlike in the wire-line connected networks, 
wireless connections are subject to various sources of imperfection, e.g., fading channels, communication delay and (non-)linear distortion on the transmit signal, which can directly impact the training process \cite{zeng2020federated}. This has led to a rich line of research work aiming to integrate the \ac{fl} framework into wireless systems \cite{amiri2020federated, nazer2007computation,liu2020over, xu2021learning}.

The state-of-the-art techniques  are mainly categorized into two streams: one that treats communication and computation separately via the \textit{transmission-then-aggregation policy} \cite{amiri2020federated}, and the other that addresses both tasks jointly by invoking the idea of \textit{analog function computation} \cite{nazer2007computation,liu2020over}. 
The latter approach, often referred to as \ac{ota-fl}, exploits the superposition applied by the wireless multiple-access channel on the input signals to realize the model aggregation directly over the air \cite{yang2019federated,xu2021learning,bereyhi2022matching}. In this work, we focus on the latter approach that has been shown to perform more efficiently in several use cases  \cite{xu2021learning}.
The key idea in \ac{ota-fl} can be explained in a nutshell as follows: the edge devices transmit their local models with proper scaling simultaneously and synchronously over the same radio resource such that the desired model aggregation is determined by the linear superposition of the uplink channel. This way the devices are not required to communicate over orthogonal resources, i.e., individual bandwidth and/or time slots. \ac{ota-fl} hence offers two key advantages over the conventional approaches based on orthogonalization of the resources; first, it reduces the computational load at the \ac{ps}, since the aggregation is done over the air. Second, it improves the resource efficiency, as it allows for non-orthogonal communication \cite{cao2021optimized}. 

{
The mentioned gains of \ac{ota-fl} comes at the expense of noisy model aggregation, due to undesired interference and noise in the channel. Nevertheless, the aggregation error can be significantly suppressed by efficient receive beamforming using large antenna-arrays at the \ac{ps} which can be realized using massive \ac{mimo} systems.
To this end, in this work, we focus on the highlighting the significance of massive MIMO-enabled OTA-FL in the presence of limited number of RF chains. Specifically, we design efficient  joint beamforming and antenna selection methods to enhance the performance of massive MIMO-enabled OTA-FL. 
}


\subsection{Related Work} 
It is apparent that \ac{ota-fl} scheme describes a trade-off: on one hand, it reduces the communication and computation costs by aggregating the global model directly over the air. On the other hand, unlike noise-free \ac{fl}, over-the-air aggregated model is perturbed as the computation is performed on a channel that experiences fading, multi-user interference, and \ac{awgn}. The calculated global model hence contains aggregation error.
Consequently, the main body of work on  \ac{ota-fl} addressed three key tasks: device scheduling, uplink coordination, i.e., device power control, and beamforming at the \ac{ps}. The goal  is to minimize the aggregation error evaluated via an error metric \cite{yang2019federated,yang2020federated,bereyhi2022matching}. 

The joint design of beamforming and device scheduling policy in \ac{mimo} settings was studied in the initial work \cite{yang2020federated}. A low-complexity design based on the matching pursuit method was later proposed in  \cite{bereyhi2022matching}. A novel unit-modulus  computation framework was proposed in \cite{wang2022edge} to reduce communication delay and implementation costs via analog beamforming. 
Low-complexity algorithms for device coordination in \ac{ota-fl} based on the \ac{mmse} and \ac{zf} methods were proposed in \cite{jahed}. {The proposed algorithms present efficient approximation of the optimal \ac{mmse} and \ac{zf} schemes using a tree-based search algorithm. }
The problem of power control for \ac{ota-fl} was investigated in \cite{cao2019optimal, zhang2021gradient, zhu2019broadband}.
{The study in \cite{cao2019optimal} proposes a joint design for device power control and the receiver beamforming at the \ac{ps}. A gradient-statistics-aware power control scheme was later introduced in \cite{zhang2021gradient} to accelerate the performance of \ac{ota-fl}.
The study in \cite{zhu2019broadband} proposed  truncated power control  for excluding the edge devices that experience deep channel fading. The authors show that the proposed algorithm provides a good balance between learning performance and aggregation error with low-latency.} 

A learning-based resource allocation algorithm for enhancing the transceiver design in \ac{ota-fl} was proposed in \cite{zou2022knowledge}. The algorithm is trained to minimize the aggregation error accumulated over all communication rounds. Considering \ac{mimo} systems with reduced complexity, the study in \cite{zhai2021hybrid} designs a low-complexity hybrid analog-digital beamforming scheme to establish \ac{ota-fl} in \ac{mimo} systems with large passive antenna arrays. The task-oriented design of \ac{irs}-aided \ac{mimo} systems was further investigated in \cite{wang2021federated}.  


\subsection{Motivation and Contributions}
Different from the existing literature, we propose efficient \textit{beamforming and antenna selection} solutions to enhance the performance of \textit{\ac{ota-fl} in massive \ac{mimo} systems}. Massive MIMO provides a promising beamforming gain at the \ac{ps} with considerable suppression of error in over-the-air aggregation step. This gain however relies on employing all the antennas at the server. From the implementation viewpoint, this means that each antenna should be allocated by an individual \ac{rf}-chain, i.e., power amplifier and analog to digital converter. This can pose a high hardware cost and complexity to the system, making it infeasible in practice. 
Subsequently, our motivation follows from the low-cost low-complexity approach of antenna selection that mitigates the cost and the complexity  of massive \ac{mimo}-enabled OTA-FL systems\cite{asaad2018massive}. 
\color{black}
\color{black} To this end, our main contributions are summarized below:
\begin{itemize}
	\item We consider a massive \ac{mimo}-enabled \ac{ota-fl} system in which a \ac{ps} equipped with a large number of antennas collaboratively trains a common model via a large number of edge devices.  Due to limited RF chains in practice, the \ac{ps} performs antenna selection, i.e., selects a subset of its antennas and beamforming jointly. 
	\item To tackle the design  problem, we first propose an algorithm by invoking the recent \ac{pdd} technique developed for optimization problems with unit-modulus and/or selection constraints \cite{shi2020penalty}.  For the target problem, we derive the penalized and dual programming that addresses the joint antenna selection and beamforming via an unconstrained optimization. We then employ the \ac{ao} method to approximate the solution of this problem within a polynomial time.
	\item Utilizing the sparse nature of the design given by joint antenna selection and beamforming, we develop an alternative algorithm based on sparse recovery via the \ac{Lasso}.  
	To this end, we first derive an alternative form of the original design problem  that can be interpreted as the sparse recovery problem. We then invoke the \ac{Lasso} algorithm along with the \ac{ao} method to design a computationally-feasible algorithm for joint antenna selection and beamforming. 
	\item Targeting applications with limited computational capacity, we develop a third class of design algorithms that impose a considerably lower computational complexity on the system. Our proposed algorithm invokes the iterative soft-threshold-based technique to bypass the linear programming task in the \ac{Lasso}-based scheme. Our investigations show that the proposed approach can closely track the two other algorithms at a considerably lower complexity. 
	\item We evaluate the efficiency of the proposed schemes 
	through numerical simulations. We investigate the image classification problem over the FMNIST and CIFAR-10 datasets via \ac{ota-fl} with both the \ac{iid} and non-\ac{iid} data distributions. Our investigations depict that all the three proposed techniques outperform the classical benchmarks such as random selection, greedy antenna selection \cite{gharavi2004fast} and all-antenna selection. In scenarios with small number of \ac{rf}-chains the \ac{Lasso}-based approach can improve the test accuracy by 5 to 10 percent as compared with the \ac{pdd}-based scheme at the expense of higher complexity imposed by algorithm tuning. The iterative soft-threshold-based approach, on the other hand, leads to a slightly higher aggregation error while enjoying a significantly lower computational complexity than the other two techniques.\color{black}
 
	
\end{itemize}

The remainder of this paper is structured as follows. Section II introduces the problem formulation and system model. Section III proposes the \ac{pdd}-based method. Sections IV and V present the \ac{Lasso} and its low-complex version, i.e., FISTA algorithm, respectively. Section VI provides the simulation results and Section VII draws the conclusion.

\subsection{Notation}
Scalars, vectors and matrices are represented with non-bold, bold lower-case, and bold upper-case letters, respectively. The transposed and the transposed conjugate of $\mH$ are denoted by $\mH^{\trp}$ and $\mH^{\her}$, respectively. $\mI_N$ and $\mone_N$ are the $N\times N$ identity and all-one matrices, and $\Vert \bx \Vert_\ell$ and $\Vert \bx \Vert_0$ denote the  $\ell$-norm~and zero norm of $\bx$, respectively. The sets $\setR$ and $\setC$ refer to~the real axis and the complex plane. $\mathcal{CN}\brc{\eta,\sigma^2}$ represents the complex Gaussian distribution with mean $\eta$ and variance $\sigma^2$. $\Ex{.}{}$ denotes the expectation of an input variable. For the sake~of~brevity, $\set{1,\ldots,N}$ is shortened to $\dbc{N}$. Furthermore, $\odot$ represents entry-wise product and $o_N$~is $N$~dimensional vector of all ones. $\text{Sgn}(x)$ is the sign function which returns the sign of the real number~$x$.

\section{System Model and Problem Formulation}
We consider \ac{ota-fl} in a wireless network with $K$ single-antenna edge devices. The devices are coordinated by a multi-antenna \ac{ps} to cooperatively address a common learning task, e.g., training a \ac{nn}, over their distributed local datasets using the federated averaging scheme. The \ac{ps} is equipped with $N$ receive antenna elements. 
Let $\set{1,2, \cdots, K}$ denote the set of edge devices participating in the learning task. Without loss of generality, we assume that the devices are to address a supervised learning task.  Each edge device has a local training dataset, denoted by $\mathcal{D}_k=\set{ \left(\mathbf{u}_{k,i}, v_{k,i}\right)}_{i=1}^{\vert \mathcal{D}_k \vert}$ with $\mathbf{u}_{k,i}$ and $ v_{k,i}$ representing the $i$-th feature vector and its corresponding label, respectively. The \textit{global dataset} is further defined as the union of all local datasets, i.e.,   $\mathcal{D}=\cup_{k=1}^K \mathcal{D}_k$. The ultimate goal in this problem is to train a learning model by minimizing the loss function $F\brc{\bomega}$ determined over the global dataset as
\begin{align}
	F(\boldsymbol{\omega} \vert \maD) \triangleq \sum _{k=1}^K \frac{|\mathcal{D}_k|}{|\mathcal{D}|}  F_k(\boldsymbol{\omega} \vert \mathcal{D}_k),
	\label{Loss}
\end{align}
where $F_k(\boldsymbol{\omega} \vert \mathcal{D}_k)$ is the local loss function determined over the local dataset of device $k$ as
\begin{align}
	F_k(\boldsymbol{\omega} \vert \mathcal{D}_k) \triangleq \frac{1}{|\mathcal{D}_k|} \sum _{ \brc{\buu_k,v_k}\in\maD_k } \ell \brc{\boldsymbol{\omega}\vert \mathbf{u}_{k}, v_{k}},
	\label{loss:loc}
\end{align}
where the sample-wise loss $\ell(\boldsymbol{\omega} \vert \mathbf{u}_{k}, v_{k})$ determines the difference between the label learned by the feature vector $\mathbf{u}_{k}$ and the true label $v_{k}$. 

For model training, the general \ac{sgd} method is considered. Consequently, the global model at the \ac{ps} is updated by averaging the gradients of the local loss functions at model parameters of the last communication round\footnote{For ease of presentation, we assume that in each communication round we update only one epoch. This is however not necessarily the case, as in practice each communication round can include multiple local training epochs.}, i.e., $\boldsymbol{\varsigma}_k \dbc{t} = \nabla F_k\brc{\boldsymbol{\omega} \dbc{t-1}\vert \maD_k}$. The updated global model is then broadcast to the participating edge devices.
In this paper, we assume that the \ac{ps} estimates the \ac{csi} accurately in the uplink channel training phase, such that the estimation error is negligible. The \ac{csi} acquisition is updated at the beginning of each channel coherence  interval that is much larger than a symbol duration. 
In the forthcoming sections, we illustrate each stage of the \ac{ota-fl} scheme in greater detail.

\subsection{Over-the-Air Model Sharing}
At the beginning of each communication round, the \ac{ps} shares the global model $\bomega$ updated in the previous round with the devices. Each device determines its local model denoted with $\boldsymbol{\varsigma}_k= \nabla F_k(\boldsymbol{\omega}\vert  \mathcal{D}_k)$. For sake of brevity, we drop the  index $t$ of the communication  round. We further focus on transmission in a single time-frequency interval in which we transmit a single model parameter $\varsigma_{k,j}$ that is a particular entry of $\boldsymbol{\varsigma}_k$. We hence drop the index $j$ and represent the particular model parameter by $\varsigma_k$.  
To share $\varsigma_k$ with the \ac{ps}, device $k$ applies a channel-dependent scaling coefficient $b_k$. Hence, its transmit signal is given by $x_k =b_k \varsigma_k$.
We assume that prior to transmission, the local models are normalized and centralized properly such that the parameters ${\varsigma_k} $ for $k \in \set{1,\ldots,K}$ are uncorrelated with zero mean and unit variance, i.e., $\Ex{\boldsymbol{\varsigma} \boldsymbol{\varsigma}^\her }{} = \mI_K$ with 
$\boldsymbol{\varsigma}=[\varsigma_1, \ldots, \varsigma_K]^\trp$ \cite{lee2020bayesian}. The scalar $b_k$ is subject to the transmit power constraint $\vert b_k \vert^2 \leq P$ for $P > 0$. 

The devices communicate over a fading Gaussian \ac{mac}. The signal arrived at the \ac{ps} array-antenna is hence given by
\begin{align} \label{eq:received-signal}
	\by_{\rm R}=\sum_{k=1}^{K} \mh_k x_k +\bn_{\rm R}, 
\end{align}
with $\mh_k \in \setC^N$ being the uplink channel coefficient of device $k$ and $\bn_{\rm R} $ denotes the \ac{awgn} process with mean zero and variance $\sigma^2$, i.e., $\bn_{\rm R} \sim \mathcal{CN}\left(0, \sigma^2 \mI_N\right)$. By defining the uplink channel matrix of the network as $\mH=[\mh_1, \ldots, \mh_K],$ the arrived signal in \eqref{eq:received-signal} can be compactly written~as
\hspace{-0.5cm}
\begin{align}
	\by_{\rm R}=\mH \mB \boldsymbol{\varsigma} +\bn_{\rm R}	,
\end{align}
where ${\mathbf B}={\mathrm{Diag}}\{b_1,\ldots,b_K\}$.  

\subsection{Antenna Selection}
The \ac{ps} is equipped with $L < N$ \ac{rf}-chains. It hence selects a subset of $L$ antenna elements in its array antenna to be active during the uplink transmission. The signal received by the \ac{ps} at the digital base-band domain can hence be written as follows:
\begin{align}
	\by = \mathbf{S} \by_{\rm R} = \mathbf{S} \mH \mB \boldsymbol{\varsigma} + \bn,
\end{align}
where $\bn=\mS \bn_{\rm R}$ captures the \ac{awgn} process on the active antennas. The matrix $\mS$ represents the switching network and is defined as $\mS={\mathrm{Diag}}\{\mathbf s\}$ with $\mathbf{s} \in \set{0,1}^N$ being the antenna selection vector whose entry $n$ for $n\in\dbc{N}$ reads $s_n = 1$ if antenna $n$ is set active, and $s_n = 0$ otherwise. As a result, we can write 
$\mathrm{tr} \set{\mS} = \bo_N^\trp \bss =   \sum_{n=1}^{N} s_n = L$. 

\subsection{Over-the-Air Model Aggregation}
The ultimate goal of the \ac{ps} is to combine the local models according to a predefined strategy specified by the \ac{fl} scheme. This means that in each uplink transmission, the \ac{ps} aims to determine the \textit{aggregated model} $\theta =\sum_{k=1}^{K} \phi_k \varsigma_k,$
where $\theta=\nabla F(\boldsymbol{\omega}) $ and  ${\phi_k}$ is a predefined weighting coefficients, i.e., ${\phi_k}=\frac{|\mathcal{D}_k|}{|\mathcal{D}|}$. To this end, the \ac{ps} invokes over-the-air computation approach and estimates the aggregated model directly from the received signal via linear beamforming, i.e., it finds the estimate of $\theta$ as $\hat{\theta} = {\bmm}^\her \by$,
for some linear receiver ${\bmm}\in\setC^N$.   Then, the \ac{ps} updates the model parameter vector with a proper step-length $\gamma$ as $\boldsymbol{\omega}^{t+1}=\boldsymbol{\omega}^t-\gamma \hat{\theta}$.

Since communication is carried out over a noisy network, the estimated global  model $	\hat{\mathbf{\theta}}$ contains some error compared to the desired global model $\theta$. This error is called \textit{aggregation error} and can be quantified via various distortion metrics. In the sequel, we invoke the information-theoretic notion of \ac{mse} to quantify the aggregation error. Hence, the error is given~by
	\begin{align}
	\epsilon \brc{ \bmm, \bss, \mB} &= \Ex{\vert {\theta-\hat{\theta} } \vert^2}{}
	\stackrel{\star}{=} \Vert {\mathbf m}^\her \mathbf S\mH\mB -\boldsymbol{\phi}^\her\Vert^2+\sigma^2 \Vert {\mathbf m}^\her \mathbf S\Vert^2,
	\label{eq:eps}
\end{align}
where $\boldsymbol{\phi}= \dbc{\phi_1, \ldots, \phi_K}$ and $\star$ is derived using the statistics of the local models and \ac{awgn}. As indicated in the formulation, the aggregation error is in general a function of linear receiver $\bmm$, switching matrix $\mS = \Diag{\bss}$ and the transmission scaling factors $\mB$. 
\subsection{Joint Design Problem}
The main design problem is hence to find $\bmm$, $\mS$ and $\mB$, such that the aggregation error is minimized subject to the edge transmit power constraints. 
This main design problem of this setting is mathematically formulated as 
	\begin{align}\label{P_1}
		\minSub{
			\epsilon \brc{ \bmm, \mS, \mB}
		}{
			\bmm, \bss, \mB 
		}{
			C_1: \bss \in \set{0,1}^N,\\
			C_2: {\bo_N}^\trp {\bss} = L,\\
			C_3: \abs{b_k}^2\leq P \text{ for } k\in[K].
		}{\maP_1}
	\end{align}
This optimization problem is in general hard to be addressed due to two main reasons: firstly, the objective function is~non-convex, and secondly, the antenna selection constraint makes the problem  NP-hard integer programming problem.  
As a result, finding the exact solution of \ref{P_1} is practically infeasible, and sub-optimal approaches for efficient approximation of the optimal design are required. In the following sections, we develop three computationally-feasible approaches based on the \ac{pdd}, \ac{Lasso} and soft thresholding methods. Both \ac{Lasso}-type and soft-thresholding algorithms work based on sparse recovery techniques and perform well at high sparsity settings, i.e., extremely few number of antennas at the \ac{ps} are active. However, with increasing the number of selected antennas at the \ac{ps} the \ac{pdd} approach performs superior.

\section{Algorithm I: A PDD-Based Method}
We start with a \ac{pdd}-based algorithm. This algorithm approximates the solution of \ref{P_1} by a two-tier iterative scheme that is derived by invoking the \ac{pdd} method. The derivations follow three major steps:
\begin{enumerate}
	\item We invoke the \ac{pdd} approach to tackle the discrete nature of the antenna selection constraint, and transform \ref{P_1} to a variational form that can be effectively relaxed.
	\item The variational problem is converted to a penalized form whose penalty is proportional to the selection constraints.
	\item Using   \ac{ao}, we develop a two-tier iterative scheme to find an approximated solution of the penalized problem.
\end{enumerate}
In the sequel, we go through each step in greater detail.

 \subsection{Finding Variational Form}
To drop the binary constraint of antenna selection, i.e., $C_2$ in \ref{P_1}, \ac{pdd} suggests introducing the auxiliary vector 
$\overline{\bss} = \dbc{ \overline{s}_1,\ldots,\overline{s}_N}^\trp$
to the optimization and let its entries satisfy the following two constraints: 
\begin{inparaenum}
	\item ${\overline s}_n=s_n$, and 
	\item $s_n\left(1-{\overline s}_n\right)=0$
\end{inparaenum}
for $n\in\dbc{N}$. We thus can equivalently find the solution of \ref{P_1} by solving the following optimization
\begin{align}\label{P_2}
	\minSub{
		\epsilon \brc{ \bmm, \bss, \mB}
	}{
		\bss, \overline{\bss}, \mB, \bmm
	}{
		C_1: {\overline s}_n=s_n \text{ and } s_n \brc{1-{\overline s}_n } = 0 \text{ for } n\in[N],\\
		C_2: \bo_N^\trp \bss = L,\\
		C_3: \abs{b_k}^2 \leq P \text{ for } k\in[K].
	}{\maP_2}
\end{align}
This follows directly from the equivalency of the constraint $C_1$ in the problems \ref{P_1} and \ref{P_2}. This equivalent form enables us to relax the antenna selection constraint more effectively.

\subsection{Deriving a Penalized Form}
\label{sec:Penalty}


As the second step, \ac{pdd} suggests transforming \ref{P_2} into a penalized form whose penalty includes all selection constraints, i.e., constraints $C_1$ and $C_2$ in \ref{P_2}. This penalized form is given by the \ac{al} dual form of \ref{P_2} that is 
\begin{align}\label{P_3}
	\minSub{
		\epsilon \brc{ \bmm, \bss, \mB}  + f_\rho \brc{\bss,\overline{\bss} \vert  \blambda} + h_\rho \brc{\bss,\overline{\bss} \vert  \bmu} +  g_\rho \brc{\bss,\overline{\bss} \vert \beta},
	}{
		\bss, \overline{\bss}, \mB, \bmm
	}{
	\abs{b_k}^2 \leq P \text{ for } k\in[K],
	}{\maP_3}
\end{align}
where the penalty terms $f_\rho \brc{\bss,\overline{\bss} \vert \blambda}$, $h_\rho \brc{\bss,\overline{\bss} \vert  \bmu} $ and  $g_\rho \brc{\bss,\overline{\bss} \vert \beta} $ are associated to the constraints $C_1$ and $C_2$ in \ref{P_2}, respectively, and are defined as follows:
\begin{subequations}
\begin{align} 
	f_\rho \brc{\bss,\overline{\bss} \vert \blambda} & = \frac{1}{2\rho} 
	\sum_{n=1}^{N} 
	\dbc{ \left(s_n-{\overline s}_n+\rho\lambda_n\right)^2 - \rho^2 \lambda_n^2 },  \label{eq:frho}\\
	h_\rho \brc{\bss,\overline{\bss} \vert \bmu} & = \frac{1}{2\rho} 
	\sum_{n=1}^{N} 
	\dbc{ \left(s_n\left(1-{\overline s}_n\right)+\rho\mu_n\right)^2 - \rho^2\mu_n^2 },  \label{eq:hrho}\\
	g_\rho \brc{\bss,\overline{\bss} \vert \beta} &=\frac{1}{2\rho}  \dbc{
	\left(\bo_N^{\mathsf T}{\mathbf s}-L+\rho\beta\right)^2- \rho^2\beta^2}. \label{eq:grho}
\end{align}
\end{subequations}
In these terms, $\rho>0$ is referred to as the \textit{penalty parameter}. The vectors $\bmu = \dbc{\mu_1,\ldots,\mu_N}^\trp$ and  $\blambda = \dbc{\lambda_1,\ldots,\lambda_N}^\trp$
with $\mu_n, \lambda_n \in \setR$ in the first and second penalty, and the scalar $\beta \in \setR$ in the third one being the \textit{dual variables}.

In principle, the solution of \eqref{P_2} is found by solving the dual problem in \eqref{P_3} for an arbitrary $\rho > 0$ and then taking the limit, when $\rho$ goes to zero \cite{shi2020penalty}. \ac{pdd} suggests to approximate this limit by forming an embedded double loop structure \cite{shi2020penalty}: Starting with an initial penalty parameter and dual variables, the inner loop utilizes \ac{ao} to solve \ref{P_3}. The solution is then fixed for the outer loop that iteratively updates either the dual variables or the penalty parameter depending on the constraint violation. The analyses in \cite{shi2020penalty} show the convergence of this algorithm to a KKT point. In the next subsection, we derive the update rules for the inner and outer loops.

%

\subsection{Inner and Outer Loops}
\label{sec:loops}
The algorithm iterates in a two-step manner. This means that in each iteration of the outer loop, the inner loop runs multiple iterations. We start with the inner loop assuming that the outer loop is at iteration $t$. We denote the penalty parameter of this iteration with $\rho^{(t)}$ and the dual variables with $\beta^{(t)}$, $\bmu^{(t)}$ and $\blambda^{(t)}$. The inner loop treats these variables as fixed and approximates the solution of \ref{P_3} via \ac{ao}. 
In other words, for the given penalty parameter and dual variables of iteration $t$, the inner loop minimizes the objective of \ref{P_3} marginally over $\bmm$, $\mB$, $\overline{\bss}$ and ${\bss}$ then iterates among the marginal solutions until it~converges. 
The marginal problems of \ref{P_3} describe standard quadratic optimizations whose solutions can be found within polynomial time. In the sequel, we discuss these marginal problems and derive the update rules of the inner loop:
\begin{enumerate}
	\item To update the \textit{linear receiver}, we solve \ref{P_3} with respect to $\bmm$ while treating all other variables fixed. Problem \ref{P_3} in this case deduces to the following unconstrained optimization
	\begin{align}\label{Optimization_Filtering}
		\min_{\bmm\in\setC^{N}} \frac{1}{2} \bmm^\her \mA \bmm -  \Re \set{ \bmm^\her \baa },
	\end{align}
where $\baa =\mS \mH \mB {\bm\phi}$ and $\mA$ is defined as
$\mA = \mS \mH \mB \mB^\her \mH^\her \mS^\her + \sigma^2 \mS \mS^\her.$
It is readily seen that $\mA\succeq 0 $. The problem in \eqref{Optimization_Filtering} features an unconstrained quadratic program that can be~solved via standard interior point methods \cite{boyd2004convex}. The solution is not generally found in a closed-form, due to the fact that $\det \mA$ can be zero\footnote{This follows from the fact that $\mS$ is a diagonal matrix whose entries are zero and one.}. The problem is however readily solved via the command \texttt{quadprog} in MATLAB by rewriting it in an augmented form.
\item The \textit{transmit scalars} are updated by solving \ref{P_3} for ${\mathbf B}$.  
This marginal problem breaks into $K$ parallel sub-problems with the $k$-th one being
\begin{align}\label{Optimization_Power}
		\min_{b_k \in \setC }  \frac{\delta_k}{2} \abs{b_k}^2 - \Re \set{ b_k^{*} \varepsilon_k }  \; \subto \abs{b_k}^2 \leq P,
\end{align}
where $\delta_k={\mathbf h}_k^{\mathsf H}{\mathbf S}^{\mathsf H}{\mathbf m}{\mathbf m}^{\mathsf H}{\mathbf S}{\mathbf h}_k \geq 0$ and $\varepsilon_k \in{\mathbbmss C}$ is the $k$-th diagonal entry of ${\mathbf H}^{\mathsf H}{\mathbf S}^{\mathsf H}{\mathbf m}{\bm\phi}^{\mathsf H}$. This is a convex quadratic optimization with a quadratic constraint. By introducing the Lagrange multiplier $\lambda$, we define the dual Lagrangian function associated as 
	\begin{align}\label{QCQP_KKT}
		{\mathcal L} \brc{b_k,\lambda} = \frac{\delta_k}{2} \abs{b_k}^2- \Re\set{ b_k^{*}\varepsilon_k } + \frac{\lambda}{2} \brc{ \abs{b_k}^2-P}.
	\end{align}
The first-order optimality condition for \eqref{QCQP_KKT} with respect to $b_k$ yields
$\delta_k b_k-\varepsilon_k+\lambda b_k=0$.
It hence follows that $b_k^\star = \bar{b}_k \brc{\lambda^\star}$, where $\bar{b}_k \brc{\lambda} = \frac{\varepsilon_k}{\delta_k+\lambda}.$
The optimal Lagrange multiplier $\lambda^\star$ is further found, such that the complementary slackness condition of the power constraint is satisfied. In particular, if $\abs{\bar{b}_k\brc{0}}^2\leq P$, then, we have $\lambda^\star=0$ and thus $b_k^{\star}=\bar{b}_k \brc{0}$. Otherwise, we have $\abs{\bar{b}_k \brc{\lambda^\star}}^2=P$ and hence $b_k^{\star}=\sqrt{P} \exp\set{\mathrm{j} \angle\varepsilon_k}$.~The~update rule is therefore given by
{
	\begin{align}\label{eq:optimalb}
		b_k^{\star}=
		\begin{cases}
			\dfrac{\displaystyle \varepsilon_k}{ \displaystyle \delta_k} &\abs{\dfrac{ \displaystyle\varepsilon_k}{\displaystyle \delta_k}}\leq P \vspace*{2mm}\\
			\dfrac{\displaystyle \varepsilon_k}{ \displaystyle \abs{\varepsilon_k} } \sqrt{P} &{\text{otherwise}}
		\end{cases}.
	\end{align}
}
\item To update the \textit{auxiliary vector} $\overline{\bss}$, we solve the marginal optimization over $\overline{\bss}$ while fixing the remaining variables. This is an unconstrained quadratic optimization. Similar to the second marginal problem, this optimization breaks into $N$ parallel sub-problems with the $n$-th sub-problem being expressed as follows:
\begin{equation}\label{Problem_Auxiliary}
	\min_{\overline{s}_n\in{\mathbbmss R}} \frac{\bar{A}_n}{2}  \overline{s}_n^2-  \bar{C}_n \overline{s}_n,
\end{equation}
where $\bar{A}_n=1+s_n^2$, and 
$\bar{C}_n=s_n^2+ \brc{1 + \rho^{(t)} \mu_n^{(t)} } s_n + \rho^{(t)} \lambda^{(t)}_n.$
The optimization problem in \eqref{Problem_Auxiliary} is a scalar quadratic program whose optimal solution is given by
${\overline{s}}_n^{\star}= \frac{\bar{C}_n}{\bar{A}_n}.$
\item The \textit{selection vector} $\bss$ is updated by solving \ref{P_3} marginally in terms of $\bss$. Similar to the marginal problem of the auxiliary vector $\overline{\bss}$,  this marginal problem breaks into $N$ parallel sub-problems, where the $n$-th sub-problem is given by
\begin{equation}\label{Beam_Selection_Matrix_Opt}
	\min_{{s}_n\in\setR} \frac{A_n}{2} {s}_n^2- C_n {s}_n,
\end{equation}
where $A_n= \dbc{\mQ}_{n,n}$ and 
$C_n= {u_n} - \sum_{n' = 1, n'\neq n}^N  \Re \set{  \dbc{\mQ}_{n',n} s_{n'}}.$
Here, the matrix ${\mathbf Q} \in \setC^{N\times N}$ is defined as follows:
\begin{align}
\mQ = 
	 \mM^\her {\mathbf H}{\mathbf B}{\mathbf B}^{\mathsf H}{\mathbf H}^{\mathsf H} \mM +\sigma^2\mM^\her \mM + \frac{1}{2\rho^{(t)}}  \left(\mone_N +{\mathbf I}_N\right)
	+
	 \frac{1}{2\rho^{(t)}}  \brc{\mI_N - \bar{\mS}}^2,
\end{align}
for $\mM = \mathrm{Diag} \set{\bmm}$ and  $\bar{\mS} = \mathrm{Diag} \set{\bar{\bss}}$. The scalar $u_n$ further denotes the $n$-th entry of $\buu$ being defined as follows:
	\begin{align}
		\buu &=  \Re \set{{\mathbf q}}
		- \frac{1}{2 \rho^{(t)} }
		\dbc{
		\left(\rho^{(t)} \beta^{(t)} - L \right)\bo_N + \left(\rho^{(t)} {\bmu^{(t)}} - \overline{\mathbf s}\right) 
		+ \rho^{(t)} \left({\bo_N} - \overline{\mathbf s} \right) \odot {\blambda^{(t)} }
	},
	\end{align}
where the vector $\mathbf q$ that is given by
	\begin{align}
	{\mathbf q} = \mathrm{diag} \set{{\mathbf H}{\mathbf B}{\bm\phi}{\mathbf m}^{\mathsf H} }.\label{eq:q}
\end{align}
It is worth noting that ${\mathbf Q}$ is a symmetric positive definite matrix; hence its diagonal entries are positive. This means that the optimization problem \eqref{Beam_Selection_Matrix_Opt} is a standard convex problem whose solution is given by
$	{{s}}_n^{\star}=\frac{{C}_n}{{A}_n}.$
\end{enumerate}



The inner loop alternates among the above four steps until it converges to a joint solution $\brc{\bmm^\star,\mB^\star, \bar{\bss}^\star,\bss^\star}$ for the given penalty parameter and dual variables in outer-loop iteration $t$. 

We next focus on the outer loop: getting back to the connection between the problems \ref{P_3} and \ref{P_2}, the direct approach for design of the outer loop is to update the dual variables considering the dual maximization problem. Invoking the steepest descent method, this is addressed by\footnote{See also \cite[Table \uppercase\expandafter{\romannumeral1}, Line 4]{shi2020penalty}.}
\begin{subequations}\label{Dual_Variable_Update}
	\begin{align}
		\beta^{(t+1)}&=\beta^{(t)}+\frac{\bo_N^{\mathsf T}{\mathbf s}-L}{2 \rho^{(t)}},\\
		\lambda_n^{(t+1)}&=\lambda_n^{(t)}+\frac{ {\overline s}_n-s_n }{2 \rho^{(t)}},\\
		\mu_n^{(t+1)}&=\mu_n^{(t)}+\frac{s_n\left(1-{\overline s}_n\right) }{2 \rho^{(t)}},
	\end{align}
\end{subequations}
for $n \in \dbc{N}$. In principle, the outer loop should iterate till both the primal and dual variables converge. In this case, we need a third-level tier to take care of the limit of $\rho$ going to zero, i.e., another loop that gradually sends the penalty parameter to zero with a certain step size; see the discussions in Section~\ref{sec:Penalty} where we illustrated the connection between \eqref{P_2} and \eqref{P_3}. \ac{pdd} invokes a trick to merge these two outer tiers, i.e., the loop that updates the dual variables and the one that takes the limit of penalty parameter going to zero, into a single loop by alternatively updating both the dual variables and the penalty parameter. To this end, at the beginning of each outer iteration, we evaluate the constraint violation of the converging solution of the inner loop by determining the \textit{violation metric} $h$ that is defined as
\begin{align}\label{Constraint_Vio}
	h = \max_{ n \in \dbc{N} }\set{
		|\bo_N^{\mathsf T}{\mathbf s}-L|,
		\left|{\overline s}_n-s_n\right|,\left|s_n\left(1-{\overline s}_n\right)\right|
	}.
\end{align}
This metric determines the maximum absolute deviation of the solution found by the inner loop from the equality constraints $C_1$ and $C_2$ in the variational optimization problem \ref{P_2}. We then compare the violation metric with a threshold $h_{\rm Th}$:
\begin{itemize}
	\item If the violation metric is smaller than the threshold, i.e., $h < h_{\rm Th}$, we consider the primal solution to give a good approximation and update the dual problem via \eqref{Dual_Variable_Update}. In this case, the penalty parameter is kept unchanged.
	\item If the violation metric is larger than the threshold, we skip the update of the dual variables and reduce the penalty parameter by multiplying it with a factor $\kappa < 1$, i.e., $\rho^{{(t+1)}} = \kappa \rho^{(t)}$.
\end{itemize}
At the end of the iteration, we further set the threshold to be the reduced version of $h$ by the same factor $\kappa$, so that the violation metric converges to zero as the outer loop converges, i.e., we set $h_{\rm Th} \leftarrow \kappa h$. The two-tier \ac{pdd}-based algorithm is summarized in Algorithm \ref{Algorithm2}. Following the analysis of \cite{shi2020penalty}, this algorithm is guaranteed to converge to a set of stationary solutions to the problem \eqref{P_1}. 

\begin{algorithm}[t]
	\caption{Algorithm I: \ac{pdd}-based Algorithm}
	\label{Algorithm2}
	\begin{algorithmic}[1]
		\STATE \textbf{Initialization:} Set the outer iteration index to $t=0$. Set the primary variables ${\mathbf m}$, ${\mathbf B}$, ${\mathbf{s}}$ and $\overline{\mathbf{s}}$ and the dual variables $\beta^{(0)}$, ${\bm\lambda}^{(0)}$, and ${\bm\mu}^{(0)}$ to some initial values. Set the violation metric threshold to $h_{\rm Th}^{(0)}$ and the penalty parameter to $\rho^{(0)}>0$. Choose a scaling factor $0 < \kappa < 1$.
		\REPEAT[\texttt{outer loop}]
		\REPEAT[\texttt{inner loop}]
		\STATE Update ${\mathbf m}$, ${\mathbf B}$, $\overline{\mathbf{s}}$  and ${\mathbf s}$ via \ac{ao} scheme described in Steps 1--4
		\UNTIL{convergence}
		\STATE Evaluate the constraint violation metric $h$ from \eqref{Constraint_Vio}
		\IF{$h < h_{\rm Th}^{(t)}$}
		\STATE Update the dual variable via \eqref{Dual_Variable_Update}
		\STATE Set $\rho^{(t+1)} = \rho^{(t)}$
		\ELSE
		\STATE Update the penalty parameter as $\rho^{(t+1)} = \kappa\rho^{(t)}$
		\STATE Set $\beta^{(t+1)} = \beta^{(t)}$, $\blambda^{(t+1)} = \blambda^{(t)}$ and $\bmu^{(t+1)} = \bmu^{(t)}$
		\ENDIF
		\STATE Set $h_{\rm Th}^{(t+1)}=\kappa h$ and $t \leftarrow t+1$
		\UNTIL{convergence}
	\end{algorithmic}
\end{algorithm}
\section{Algorithm II: A \ac{Lasso}-Type Algorithm}
According to \cite{bereyhi2017asymptotics,bereyhi2018stepwise,bereyhi2019glse}, antenna selection problem can be naturally formulated as a sparse recovery problem. In this section, we invoke this viewpoint and develop a low-complexity algorithm based on the well-known method of \ac{Lasso} for sparse recovery and regression \cite{tibshirani1996regression}.
\subsection{Selection via Sparse Regression}
To start with the \ac{Lasso}-type algorithm, we get back to the primal optimization problem \ref{P_1}. In this problem, we have two constraints: 
\begin{enumerate}
	\item the transmit power of each edge-device must satisfy the power constraint, and
	\item the number of selected antennas at the \ac{ps} should be $L$.
\end{enumerate}
The second constraint can be written in terms of the $\ell_0$-norm of the selection vector $\bss$ as $\norm{\bss}_0 = L$. We hence can rewrite the problem \ref{P_1} as
	\begin{align}\label{P_1_prime}
	\minSub{
		\epsilon \brc{ \bmm, \bss, \mB}
	}{
		\bmm, \bss, \mB 
	}{
		C_1: \bss \in \set{0,1}^N,\\
		C_2: \norm{\bss}_0 = L,\\
		C_3: \abs{b_k}^2\leq P \text{ for } k\in[K].
	}{\bar{\maP}_1}
\end{align}

Intuitively, the equality constraint $C_2$ can be further replaced by an inequality: in fact setting $\norm{\bss}_0 \leq L$ means that some \ac{rf}-chains are set off at the \ac{ps}, which can only degrade the performance of our system. In other words, we can intuitively claim that the solution to \ref{P_1_prime} is given by solving 
\begin{align}\label{P_4}
	\minSub{
		\epsilon \brc{ \bmm, \bss, \mB}
	}{
		\bmm, \bss, \mB 
	}{
		C_1: \bss \in \set{0,1}^N,\\
		C_2: \norm{\bss}_0 \leq L,\\
		C_3: \abs{b_k}^2\leq P \text{ for } k\in[K].
	}{{\maP}_4}
\end{align}
This claim is shown to be valid in the following theorem.

\begin{theorem}
	\label{theorem:P4}
	Let $\bss^\star$ be a solution of the optimization problem \ref{P_1_prime}. Define the maximum zero-forcing error of $\bss^\star$ to be
$E_{\rm zf} =  \Vert {\mathbf m}^\her \Diag{\bss^\star} \mH\mB -\boldsymbol{\phi}^\her\Vert_{\infty}$.
If $E_{\rm zf} \neq 0$, then $\bss^\star$ is also a solution of the optimization problem \ref{P_4}.
\end{theorem}

\begin{proof}
	The proof follows from the fact that with a non-zero zero-forcing error, there always exists a non-zero receiving gain for a non-selected antenna which reduces the aggregation error. This concludes that there does not exist a selection vector $\bss$ with $\norm{\bss}_0 < L$ that leads to a smaller error, and hence \ref{P_4} recovers the same selection vector. The details of the proof are given in Appendix \ref{Proof_th_P4}.
\end{proof}


The optimization problem in \ref{P_4} describes a sparse regression problem in which the regression error is described via the aggregation error. Following the standard regularization technique, we find the alternative regularized form of \ref{P_4} as follows:
\begin{align}\label{P_5}
	\minSub{
		\epsilon \brc{ \bmm, \bss, \mB} + \eta \norm{\bss}_0
	}{
		\bmm, \bss, \mB 
	}{
		C_1: \bss \in \set{0,1}^N,\\
		C_2: \abs{b_k}^2\leq P \text{ for } k\in[K],
	}{{\maP}_5}
\end{align}
for some regularizer $\eta > 0$. It is readily shown that there exists some $\eta_0$, such that \ref{P_5} recovers the solution of \ref{P_4} at $\eta = \eta_0$. 

\subsection{Sparse Regression via Box-Lasso}
The problem \ref{P_5} in its primitive form is a classical computational task arising in sparse recovery which leads to the \ac{np}-hard problem of integer programming. Various sub-optimal algorithms have been proposed in the literature to address this problem in a tractable manner \cite{foucart13}. The most well-known approach is the \ac{Lasso} technique which approximates the solution of the sparse regression problem via the convex minimization $\ell_1$-norm relaxation \cite{tibshirani1996regression}. In this approach, the $\ell_0$-norm is relaxed by the convex $\ell_1$-norm. Analyses in \cite{tibshirani1996regression} have shown that this relaxation guarantees the recovery of a sparse solution.

The basic form of \ac{Lasso} deals with an unconstrained problem, i.e., $s_n$ are on the real axis. For the constrained case, the method is further extended to the \textit{box-\ac{Lasso}}; see \cite{bereyhi2021detection} and the references therein for various forms of box-\ac{Lasso} and its applications. In this scheme, the constraint set is further relaxed to a convex set, referred to as \textit{box}, which includes the non-convex feasible set. Considering the constrained nature of \ref{P_5}, we use the box-\ac{Lasso} method to address the underlying sparse regression problem in \ref{P_5}. To this end, we relax the integer set $\set{0,1}^N$ to the convex box $\dbc{0,1}^N$. The box-\ac{Lasso}-based relaxed problem is hence given by
\begin{align}\label{P_6}
	\minSub{
		\epsilon \brc{ \bmm, \bss, \mB} + \eta \norm{\bss}_1
	}{
		\bmm, \bss, \mB 
	}{
		C_1: \bss \in \dbc{0,1}^N,\\
		C_2: \abs{b_k}^2\leq P \text{ for } k\in[K],
	}{{\maP}_6}
\end{align}
The optimization in \ref{P_6} still describes a non-convex joint optimization problem. Its advantage is however that its marginal optimizations, i.e., optimization over a single variable while treating the others as constants, are convex. This allows us to leverage the \ac{ao} method and obtain an efficient sub-optimal solution to the problem.

\begin{remark}
	It is worth mentioning that the regularizer $\eta$ in \ref{P_6} is a free variable which is tuned in practice, such that the design performance is optimized. We discuss this tuning task later in greater detail throughout the numerical investigations.
\end{remark}

\subsection{Iterative Algorithm Based on AO}
We next invoke the \ac{ao} method to develop an iterative algorithm for approximating the solution of \ref{P_6}. For this problem, the \ac{ao} method alternates among three marginal problems, i.e., marginal problems with respect to the receiver, to the transmission scalars  and to the selection vector. In the sequel, we derive the solution of each marginal problem:
\begin{enumerate}
	\item The marginal problem with respect to the \textit{linear receiver} $\bmm$ finds the same form as the one given in Step 1 of Section~\ref{sec:loops}. The solution is hence given by solving the quadratic optimization problem in \eqref{Optimization_Filtering}.
	\item The optimization over \textit{transmit scalars} breaks into $K$ parallel scalar optimizations with the $k$-th one being described in \eqref{Optimization_Power} in Step 2 of~Section~\ref{sec:loops}.~The solution is hence~given~by~\eqref{eq:optimalb}.
	\item The last marginal problem is to optimize the objective with respect to the \textit{selection vector} $\bss$, while treating ${\mathbf B}$ and ${\mathbf m}$ as constants. In this case, the corresponding problem is given by
	\begin{align}\label{Lasso:1}
		{\mathbf s}^\star=\argmin_{\bss\in\dbc{0,1}^N } \epsilon \brc{ \bmm, \bss, \mB} + \eta \norm{\bss}_1. 
	\end{align}
By substituting the definition of the aggregation error into \eqref{Lasso:1}, and after some derivations, we conclude that the solution is given by solving the following box-Lasso problem
\begin{align}\label{Lasso:2}
	{\mathbf s}^\star=\argmin_{\bss\in\dbc{0,1}^N } \bss^\her \mQ_{\rm lasso} \bss - 2 \Re\set{\bss^\her \mA_{\rm lasso} \boldsymbol{\phi}} + \eta \norm{\bss}_1, 
\end{align}
where $ \mA_{\rm lasso} = \mM \mH \mB $ and $ \mQ_{\rm lasso} $ is defined as
\begin{align}\label{eq:Q}
\mQ_{\rm lasso} = \mM \brc{\mH \mB \mB^\her \mH^\her} \mM^\her,
\end{align}
for $\mM = \Diag{\bmm}$. This is a convex optimization problem and can be solved directly by means of classical convex programming toolboxes, e.g., CVX in MATLAB or CVXPY for programming in Python. Alternatively, one can develop an iterative algorithm to solve \eqref{Lasso:2} with minimal computational complexity; see \cite[Chapter 15]{foucart13} and \cite{arnold2016efficient} and references therein for some instances of such algorithms.
\end{enumerate}
The above steps are alternated until the solution converges. If the converging solution of $\bss$ contains more than $L$ non-zero entries, then the selection vector is set to the $L$ largest entries~of~$\bss$.
The box-Lasso-type algorithm is summarized in Algorithm~\ref{Algorithm_Lasso}. As mentioned, to optimize the performance further, one needs to further tune the regularizer $\eta$ in this algorithm. This task is often performed adaptively via numerical approaches; see for instance \cite{leng2006note}. Analytic approaches based on large-system analysis are an alternative tuning approach which can be followed in \cite{bereyhi2021detection,bereyhi2019rls}.

\begin{algorithm}[t]
	\caption{Algorithm II: Box-Lasso-Type Algorithm}
	\label{Algorithm_Lasso}
	\begin{algorithmic}[1]
		\STATE \textbf{Initialization:} Set the iteration index to $t=0$. Set the primary variables ${\mathbf m}$, ${\mathbf B}$ and ${\mathbf{s}}$. Choose a regularizer $\eta>0$.
		\REPEAT
		\STATE Update ${\mathbf m}$ by solving \eqref{Optimization_Filtering}
		\STATE Update $b_k$ for $k\in\dbc{K}$ via \eqref{eq:optimalb}
		\STATE Update $\bss$ by solving \eqref{Lasso:2}
		\STATE Set $t \leftarrow t+1$
		\UNTIL{convergence}
		\STATE Set $\bmm$ and $\mB$ to the converging solutions
		\STATE Set $L$ largest entries of $\bss$ to $1$ and the remaining to zero  
	\end{algorithmic}
\end{algorithm}

\section{Algorithm III: An Algorithm Based on Iterative Soft-Thresholding}
Various sparse recovery techniques deal with quadratic programming \cite{foucart13}. A classical approach for the implementation of these techniques is to approximate the solution of the optimization problem via the Gauss-Seidel method. This results in an algorithm whose complexity linearly scales with the problem dimension. In this section, we invoke this approach to develop a low-complexity algorithm for joint antenna selection and beamforming from the Lasso-based algorithm.

We start from the box-Lasso approximation of the original design problem, i.e., \eqref{P_6}. Similar to Algorithm II, we address this problem via iterative \ac{ao}, where in each iteration we alternate between the marginal beamforming problem in \eqref{Optimization_Filtering} and the antenna selection problem in \eqref{Lasso:1}. The Gauss-Seidel method suggests approximating the solution of the latter problem in a step-wise fashion: in each step, we find the optimal entry $s_n$ while treating the other entries to~be~fixed. 

Considering the marginal problem in \eqref{Lasso:1}, the Gauss-Seidel method breaks the $N$-dimensional optimization into $N$ parallel scalar box-Lasso problems with the $n$-th one being
\begin{equation}\label{eq:s-lasso}
	{s}_n^{\star}=\argmin_{ 0 \leq {s}_n \leq 1}\dfrac{1}{2}\left(s_n-\frac{z_n}{2w_n}\right)^2+\frac{\eta}{2w_n}\vert s_n\vert.
\end{equation}
Here, $z_n$ and $w_n$ are described as
\begin{subequations}
	\begin{align}
		z_n&=\vert m_n \vert^2\left( \sigma^2 +\sum_{k=1}^{K}\vert b_k\vert^2 \vert h_{nk}\vert^2\right), \label{eqZ}\\
		w_n&=\Re\{\frac{q_n}{2}-\sum_{n'=1,n'\neq n}^N [\mQ_{\rm lasso} ]_{n',n}s_{n'}\}, \label{eqW}
	\end{align}
\end{subequations}
with $\mQ_{\rm lasso} $ being defined in \eqref{eq:Q}. The scalars  $m_n$ and $q_n$ denote the $n$-th entry of $\mathbf{m}$ and $\mathbf{q}$ defined in \eqref{eq:q}, respectively, and $h_{nk}$ is the entry of $\mH$ at row $n$ and column $k$. The solution of this scalar optimization is given in a closed form as
${s}_n^{\star}= \dfrac{1}{2w_n} T_\eta^S(z_n),$
where $T_\eta^S\brc{u}$ is the so-called \textit{soft-thresholding operator} defined as
\begin{align}
T_\eta^S(u)= 
\begin{cases}
	0,& \vert x\vert \leq \eta\\
	u-\eta \text{Sgn}(u),  & \vert x\vert \geq \eta
\end{cases}.
\end{align}

Using this approximative approach, the antenna selection in each iteration of the algorithm is performed by $N$ simple soft-thresholding operations: in each \ac{ao} iteration, the algorithm finds the new switching vector as ${s}_n^{(t+1)} = T_\eta^S(z_n^{(t)})/{2w_n^{(t)}} $, where $w_n^{(t)}$ and $z_n^{(t)}$ are determined from \eqref{eqZ} and \eqref{eqW} by replacing $\bmm$ and $\bss$ with their values in iteration $t$. The final algorithm based on the \ac{ao} method is summarized in Algorithm~\ref{Algorithm_fista}. We refer to this algorithm as a \ac{fista} as it invokes iterative soft-thresholding for antenna selection.

\begin{algorithm}[t]
	\caption{Algorithm II: FISTA-based Algorithm}
	\label{Algorithm_fista}
	\begin{algorithmic}[1]
		\STATE \textbf{Initialization:} Set the iteration index to $t=0$. Set the primary variables ${\mathbf m}$, ${\mathbf B}$ and ${\mathbf{s}}$. Choose a regularizer $\lambda>0$.
		\REPEAT
		\STATE Update ${\mathbf m}$ by solving \eqref{Optimization_Filtering}
		\STATE Update $b_k$ for $k\in\dbc{K}$ via \eqref{eq:optimalb}
		\STATE Determine $z_n$ and $w_n$ from \eqref{eqZ} and \eqref{eqW}, respectively, and update $\bss$ as
		\begin{align*}
			{s}_n= \dfrac{1}{2w_n} T_\eta^S(z_n)
		\end{align*}
	\vspace{-0.8cm}
		\STATE Set $t \leftarrow t+1$
		\UNTIL{convergence}
		\STATE Set $\bmm$ and $\mB$ to the converging solutions
		\STATE Set $L$ largest entries of $\bss$ to $1$ and the remaining to zero  
	\end{algorithmic}
\end{algorithm}



\section{Complexity and Convergence Analysis}
{
In this section, we first provide the complexity analysis of the proposed algorithms in previous sections and then give a short discussion on convergence of the OTA-FL Scheme.}
\subsection{Complexity Analysis of the Proposed Algorithms}

\subsubsection{PDD-based method}
In general, the \ac{pdd} algorithm is more complex as compared to standard step-wise approaches, as it iterates through two tiers. To derive the complexity of this algorithm analytically, let $I_{\rm{out}}$ and $I_{\rm{in}}$ denote the numbers of iterations in the outer loop and the inner loop, respectively. For each loop, we can approximate the complexity of each iteration with the dominant computational task of the iteration:
\begin{itemize}
	\item In the \textit{outer loop}, the dominant computational task in each iteration is the inner loop.
	\item The per-iteration complexity in the \textit{inner loop} is mainly composed of the complexity of updating the primal variables ${\mathbf m}$, ${\mathbf B}$, $\overline{\mathbf{s}}$  and ${\mathbf{s}}$: updating ${\mathbf m}$ requires solving a quadratic program whose complexity scales with $K^3$. The update of ${\mathbf B}$ needs $K$ parallel updates~each dealing with a matrix-vector multiplication of order $N$. The complexity hence scales with $KN$. Similarly, the update of $\bss$ requires $N$ parallel updates, each evaluating a matrix-vector product of complexity order $K$. The update of $\bss$ hence  scales with $KN$. Finally, $\bar{\bss}$ is updated through $N$ parallel updates of constant complexity, and its update scales linearly with $N$.
	%
\end{itemize}
Assuming that $N$ and $K$ scale proportionally, the per-iteration complexity of the inner loop can be approximated with $K^3$, and thus the overall complexity of Algorithm \ref{Algorithm2} scales with $I_{\text{out}}I_{\text{in}}K^3$ which is cubic in system dimension for fixed numbers of iterations.

Although the computational complexity of the proposed \ac{pdd}-based algorithm is feasible in practice, it is still considered to be high in many systems with limited processing capacity. 
\subsubsection{Complexity Analysis of Lasso-type Algorithm}
The computational complexity of the box-Lasso-type algorithm depends on the algorithm adapted for solving the box-Lasso problem in \eqref{Lasso:2}. Though the computational complexity varies from one implementation to another, we can consider the classical implementation based on \ac{lars} \cite{efron2004least}. In this case, the computational complexity of the box-Lasso problem scales cubically with the number of devices, i.e., $K^3$. Noting that the first and second marginal problems scale with $K^3$ and $KN$, we can conclude that the per-loop computational complexity scales with $K^3$. Assuming $I_{\rm AO}$ iterations for convergence, the algorithm imposes a computational complexity of order $I_{\rm AO} K^3$ to the system, similar to the \ac{pdd}-based scheme. 

Considering both Algorithms I and II, we note that the quadratic programming in \eqref{Optimization_Filtering} and the box-Lasso problem in \eqref{Lasso:2} are the main computational bottlenecks. In the sequel, we develop a reduced-complexity algorithm which addresses these bottlenecks by using low-complexity alternatives for these two tasks.
\subsubsection{Complexity Analysis of Iterative soft Thresholding based algorithm}
The direct implementation of the \ac{fista}-based selection deals only with matrix-vector multiplications and scales with $KN$, which is in the same order as the complexity of the marginal problem of updating $\mB$. The complexity is hence dominated by the update of $\bmm$, i.e., $K^3$. 

Assuming $I_{\rm AO}$ iterations for convergence, the algorithm imposes a computational complexity of order $I_{\rm AO} K^3$ to the system, similar to the previous schemes. Nevertheless, as the antenna selection is performed at a significantly lower complexity, the \ac{fista}-based approach performs with a constant factor faster than the other two approaches. We discuss this point in greater detail throughout the numerical investigations.

\subsection{Convergence of the OTA-FL Scheme}
In this section, we discuss the convergence of the global model parameters to the optimal solution, i.e., $\min_{\bomega} F\brc{\bomega \vert \maD}$.  The convergence in general depends on various aspects:
	\begin{inparaenum}
		\item the analytic properties of the local loss functions,
		\item the optimization algorithm used for local updates, and
		\item the joint \ac{as} and beamforming algorithm used for over-the-air model aggregation.
	\end{inparaenum}
Let's assume that the sample-wise loss function in \eqref{loss:loc}, i.e., $\ell\brc{\bomega\vert \buu_k,v_k }$, is \textit{continuously differentiable}  with respect to $\bomega$ and the global loss function $F\brc{\bomega \vert \maD}$ has a minimizer $\bomega^\star$. Moreover, we assume that the gradient of $F\brc{\bomega \vert \maD}$ is Lipschitz continuous for all $\bomega_1, \bomega_2 \in \setR^U$ and some  $L_{Lip}>0$, i.e., $\norm{\nabla F\brc{\bomega_1} -  \nabla F\brc{\bomega_2}} \leq L_{Lip} \norm{\bomega_1-\bomega_2}$ with $U$ denoting the number of model parameters. And the global loss is strongly convex with parameter $0 < \mu <L_{Lip}$, i.e., for all $\bomega, \boldsymbol{\delta} \in \setR^U$ 
\begin{align}
	F\brc{ \bomega + \boldsymbol{\delta}} \geq F\brc{\bomega} + \boldsymbol{\delta}^\trp \nabla F\brc{\bomega} + \frac{\mu}{2} \norm{\boldsymbol{\delta}}^2.
\end{align}
We further assume that the edge devices use \ac{sgd} to locally update their model parameters.

We now invoke the results of \cite{friedlander2012hybrid} which characterizes the so-called \textit{optimality gap} of the \ac{sgd} in terms of the \ac{mse} between the noisy calculation of the gradient and its true value. The optimality gap is defined to be the difference between the loss of the global model in communication round $t$, i.e., $F(\bomega\dbc{t})$, and the optimal solution, i.e., $F(\bomega^{\star})$. Using the results of \cite{friedlander2012hybrid}, we show that the optimality gap in our setting can be reduced in each iteration, and hence the \ac{fl} scheme converges to a point in the vicinity of the optimal solution. To this end, let us define the optimality gap in iteration $t$ formally as
$G\dbc{t} = \Ex{F(\bomega\dbc{t}) - F(\bomega^{\star}) }{}$.
Using Lemma 2.1. in \cite{friedlander2012hybrid}, we can bound the optimality gap in the next communication round~as 
\begin{align}
	G\dbc{t+1} \leq \brc{1-\frac{\mu}{L_{Lip}}} G\dbc{t} + \frac{1	}{2L_{Lip}} \epsilon^{(t)}\brc{\bmm, \bss, \mB}, \label{eq:upper}
\end{align}
with $ \epsilon^{(t)}\brc{\bmm, \bss, \mB}$ denoting the aggregation error in communication round $t$ determined by setting the designed $\bmm$,  $\bss$ and $\mB$ into \eqref{eq:eps}. 

We now denote the aggregation error achieved by a particular algorithm in communication round $t$ with $\epsilon\dbc{t}$. Substituting in \eqref{eq:upper}, we conclude that starting from a point with large optimality gap, the \ac{ota-fl} scheme moves towards the optimal solution as long as $G\dbc{t} > \frac{\epsilon\dbc{t}}{2\mu}$.~This~guarantees that with enough number of communication rounds the final solution determined by the \ac{ota-fl} scheme is in the vicinity of the optimal solution. The optimality gap of the converging solution moreover depends on the joint selection and beamforming algorithm: the smaller the achieved aggregation error is, the closer to the optimal solution the algorithm converges.
\color{black}
\section{Numerical Results and Discussions}
In this section, we conduct multiple experiments to evaluate performance of the proposed algorithms. We consider a multi-user network in which a set of edge devices invoke the \ac{fl} framework to learn a common model for a 10-class image classification problem.  
\subsection{Communication Settings}
\label{sec:comm_sett}
We consider a single cell of a cellular network with $K=50$ single-antenna edge devices and a \ac{ps} that is equipped with $N=128$ antennas. The maximum transmit power of each edge device is set to $P = 1$ W. The \ac{ps} is located at the center of the cell and the locations of the devices are uniformly and randomly generated within a ring whose inner radius is $R_{\mathrm{in}}=10$~m and whose outer radius is $R_{\mathrm{out}}=100$ m. The vector of channel coefficients between the device $k$ and the \ac{ps} is generated as $\mh_k=\sqrt{\rho_k^L(d_k)}\mathbf{g}_k$,~where $\rho_k^L$ models the large-scale path loss and $\mathbf{g}_k$ captures the small-scale fading. The path-loss is determined from the distance from the \ac{ps} as $\rho_k^L\left(d_k\right)=\rho_{\mathrm{ref}} \left(\dfrac{d_k}{d_{\mathrm{ref}}}\right)^{-\alpha}$ where $d_k$ denotes the distance between device $k$ and the \ac{ps}, $\alpha$ is the path loss exponent and $\rho_{\mathrm{ref}}$ is the path-loss at the reference distance of $d_{\mathrm{ref}}$. 
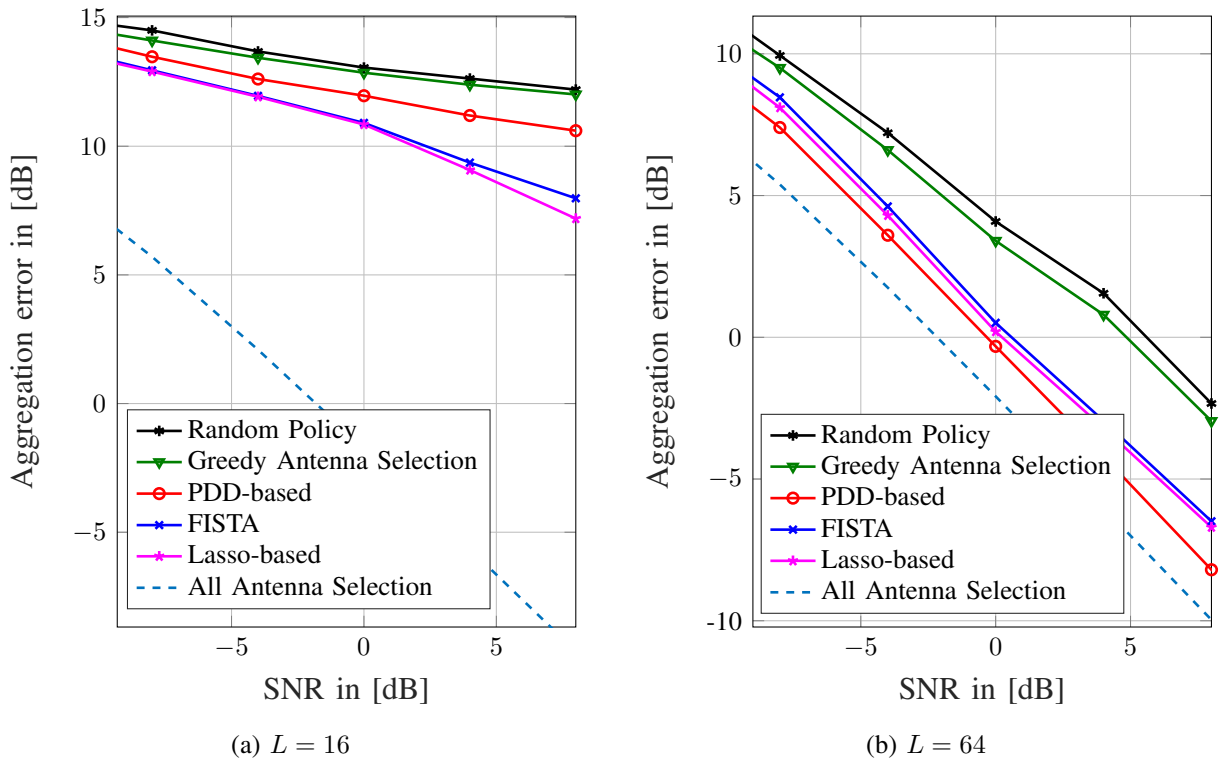
\begin{figure*}[t!]
	\centering
	\begin{subfigure}[t]{0.5\textwidth}
		\centering
		\input{fig1-16.tex} 
		\caption{$L=16$}
		\label{fig1a}
	\end{subfigure}%
	~ 
	\begin{subfigure}[t]{0.5\textwidth}
		\centering
		\input{fig1-64.tex} 
		\caption{$L=64$}
		\label{fig1b}
	\end{subfigure}
	\caption{The aggregation error versus \ac{snr} for different values of $L$.}
	\label{fig:fig1}
\end{figure*}
For the fading process, we consider the Rayleigh model. This means that the entries of $\mathbf{g}_k$ are generated \ac{iid} according to a complex Gaussian distribution with zero-mean and covariance matrix $\mathbf{R}_k$, i.e., $\mathbf{g}_k\sim \mathcal{CN} \left(0, \mathbf{R}_k\right)$. We use the Rayleigh model for correlation \cite{bjornson2017massive} meaning that the entries of covariance matrix $\mathbf{R}_k$ are set to ${\left[\mathbf{R}_k\right]}_{n, m}=u^{n-m}\left(\kappa_k\right)\xi_{n, m}^{(k)} \left(\kappa_k \right)$, where $u\left(\kappa_k\right)$ is given by $u\left(\kappa_k\right)=\exp \{j 2\pi \zeta \sin (\kappa_k)\}$ with $\kappa_k$ being the \ac{AoA} at the \ac{ps} from the $k$-th device, and $\zeta$ being the distance between two neighboring antenna elements.
Furthermore, $\xi_{n, m}^{(k)} \left(\kappa_k \right)$ is the angular spread of the \ac{AoA} at the \ac{ps} from device $k$ given by
$\xi_{n, m}^{(k)} \left(\kappa_k \right)=\exp\left\{-2\vartheta_k^2[\pi \left(n-m\right) d \cos\left(\kappa_k\right)]^2\right\}$.
Here, $\vartheta_k$ is the standard deviation of the angular spread at  device $k$ and is chosen uniformly and randomly in the interval of $\left[12, 15\right]$  throughout the simulations. The nominal values of AoAs are further calculated geometrically from the position of the devices. For the sake of comparison, we evaluate the performance for three baselines in addition to the proposed algorithms:
\begin{itemize}
	\item \textbf{Random Policy}:  An \ac{ota-fl}  scheme in which a subset of $L$ antennas at the \ac{ps} are selected randomly. In the random policy, based on the subset of selected antennas, the receive beamforming and transmit coefficients are optimized by \ac{ao} algorithm. 
	\item \textbf{Greedy Antenna Selection Policy}: A  scheme that selects the subset of $L$ antennas at the \ac{ps} which corresponds to the $L$ strongest sum of channel gains, i.e., the $L$ antennas with $L$ largest $\displaystyle \sum_{k=1}^K \abs{h_{k,n}}^2$ for $n\in \set{1,\ldots,N}$. 
	Similar to the random policy, we optimize the receive beamforming and transmit coefficients by \ac{ao} algorithm.
	\item \textbf{All Antenna Selection}: An \ac{ota-fl}  scheme in which all the antennas at the \ac{ps} are selected with the assumption that the RF chains are available for all of them. The receive beamforming and transmit coefficients are optimized alternatively. 
\end{itemize}
	\begin{figure*}[t!]
	\centering
	\begin{subfigure}[t]{0.5\textwidth}
		\centering
		\input{MSELFinal.tex}
		\caption{Standard Rayleigh fading channel}
		\label{fig2aa}
	\end{subfigure}%
	~ 
	\begin{subfigure}[t]{0.5\textwidth}
		\centering
		\input{iidL.tex} 
		\caption{Correlated Rayleigh fading channel.}
		\label{fig2bb}
	\end{subfigure}
	\caption{The aggregation error versus number of selected antennas $L$ for different channel model.}
	\label{fig:fig2}
\end{figure*}
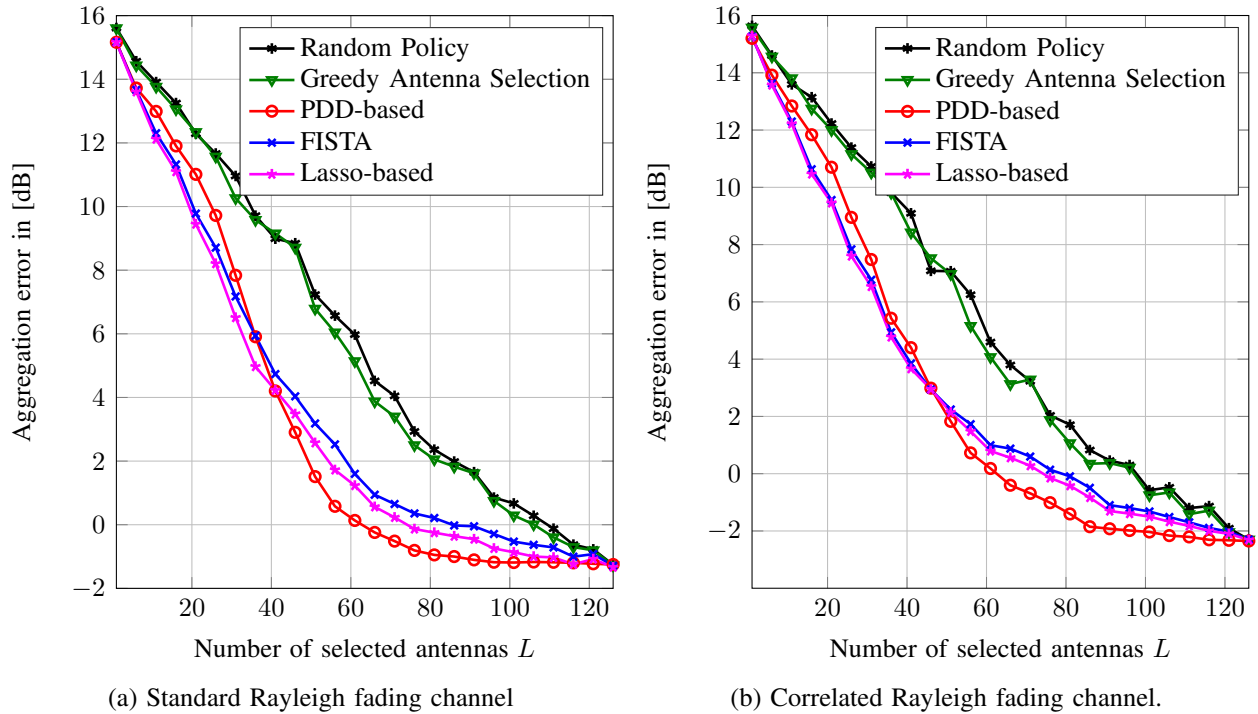
\subsection{Local Training Model Settings} 
We consider 10-class image  classification over
CIFAR-10 \cite{krizhevsky2009learning} and FMNIST \cite{xiao2017fashion} datasets. The dataset CIFAR-10 contains $10$ classes of color images with $6000$ images per class. Each class is divided into $5000$ images for training and the remaining $1000$ images are used for testing. The dataset FMNIST (Fashion-MNIST) has $10$ classes of gray-scale cloth images. There are $60,000$ training images and $10,000$ test images. 
To address the classification task, we train a \ac{cnn}. For CIFAR-10 dataset, we train a \ac{cnn}  with the widely-known VGG13 network 19 \cite{simonyan2014very}. The network consists of eight convolutional layers and two subsequent fully-connected layers. Each max-pooling layer is followed by a convolutional layer. The network is trained to minimize the empirical cross-entropy loss function in an \ac{ota-fl} manner by utilizing \ac{sgd} with momentum algorithm. The \ac{cnn} model for FMNIST dataset comprises three convolutional layers, one fully-connected layer and a classification layer with softmax function. 
We consider two data distribution scenarios:
\begin{enumerate}
\item\textbf{ \textit{\ac{iid} data distribution}} in which the overall training and testing datasets are shuffled randomly and partitioned equally and evenly among the devices.
    \item \textbf{\textit{Non-\ac{iid} distribution with label skew}} in which datasets are split with heterogeneous label distributions. Specifically, at each device, two labels are chosen at random. Then, 60\% of the datapoints in the client's dataset are chosen to be of those two labels and the remaining 40\% datapoints are chosen uniformly from the other  labels.
\end{enumerate}
We set the total number of communication rounds to be $T=5\times 10=50$ where channel  varies after every few communication rounds.

\subsection{Simulation Results}

	\begin{figure*}[t!]
	\centering
	\begin{subfigure}[t]{0.5\textwidth}
		\centering
		\input{figmnistL16.tex} 
		\caption{$L=16$}
		\label{fig2a}
	\end{subfigure}%
	~ 
	\begin{subfigure}[t]{0.5\textwidth}
		\centering
		\input{figmnistL40.tex} 
		\caption{$L=64$}
		\label{fig2b}
	\end{subfigure}
	\caption{Testing accuracy on  \ac{iid} FMNIST dataset for different number of selected antennas.}
	\label{fig:fig3}
\end{figure*}
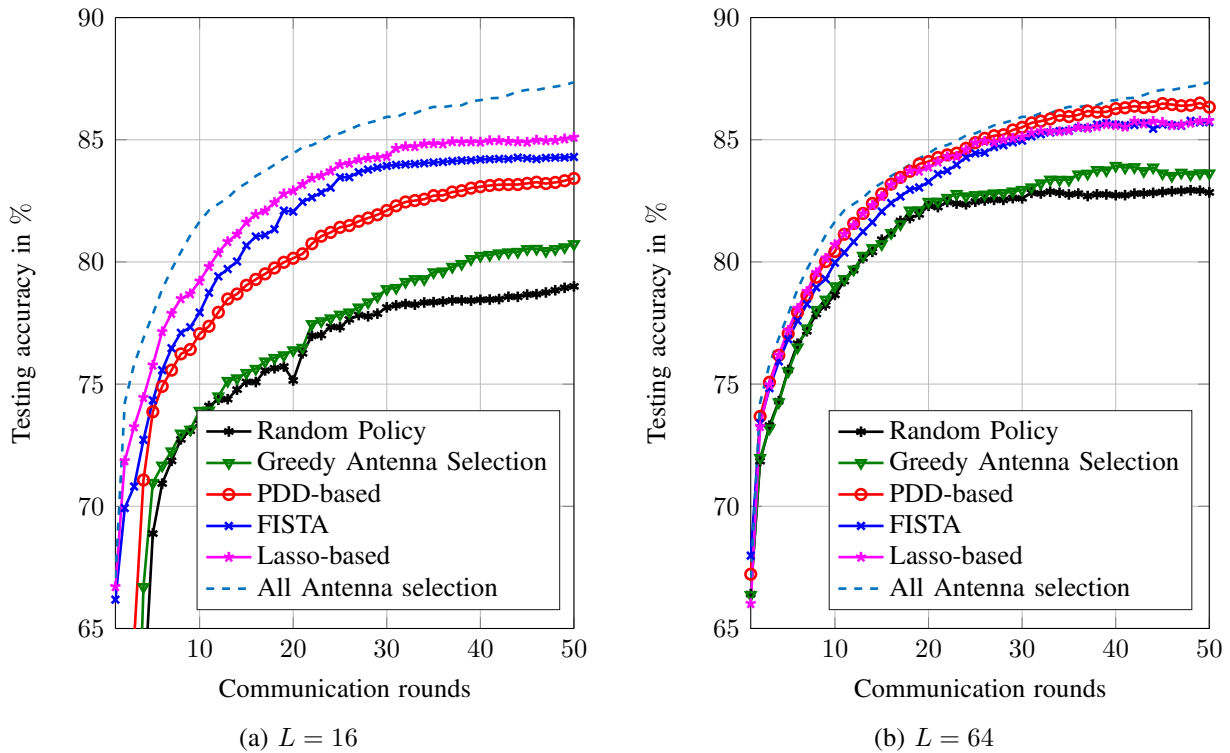

\subsubsection{Impact of Antenna Selection on Aggregation Error} We start our simulations by investigating the impact of antenna selection on aggregation error.  Fig. \ref{fig:fig1} shows the aggregation error against \ac{snr}, defined as $P/\sigma^2$ for $L=16$ and $L=64$. As expected, the figure shows that the error reduces with the increase in SNR.
	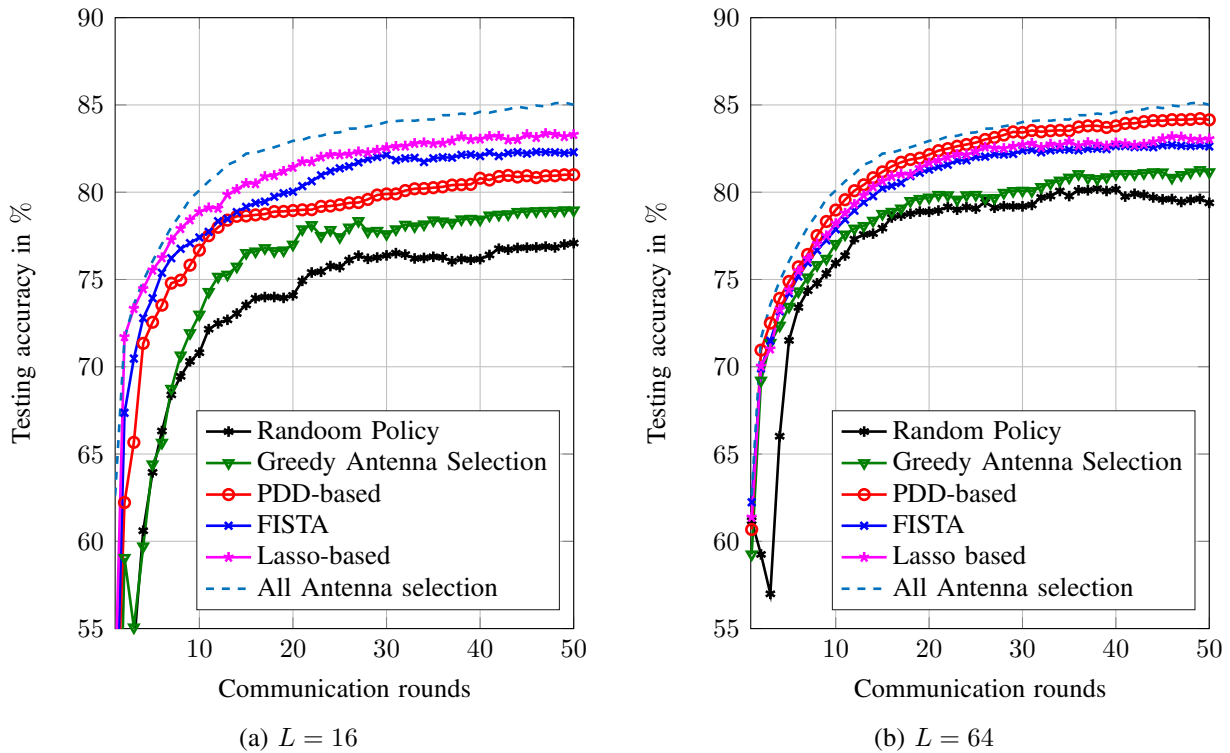
\begin{figure*}[t!]
	\centering
	\begin{subfigure}[t]{0.5\textwidth}
		\centering
		\input{fmnistn16.tex} 
		\caption{$L=16$} 
		\label{fig3a}
	\end{subfigure}%
	~ 
	\begin{subfigure}[t]{0.5\textwidth}
		\centering
		\input{fmnistn40.tex} 
		\caption{ $L=64$}
		\label{fig3b}
	\end{subfigure}
	\caption{Testing accuracy on non-\ac{iid} FMNIST dataset for different number of selected antennas}
	\label{fig:fig4}
\end{figure*}
 It is further seen in the figure that the \ac{fista}  closely tracks the Lasso-based algorithm. This is expected, as both the algorithms select the antennas based on $\ell_1$-norm minimization. The close match of both algorithms further indicates that the extra tuning and computational complexity of the Lasso-based scheme does not gain that much in terms of performance and that the \ac{fista} scheme is a good approximation for $\ell_1$-norm minimization problem. Comparing the cases of $L=16$ and $L=64$, it is further seen that the Lasso-based and \ac{fista} algorithms outperform the \ac{pdd} scheme for $L\ll N$; whereas, the \ac{pdd} scheme starts to be superior as $L$ increases substantially. This follows from the fact that both the Lasso-based and \ac{fista} inherently work based on the $\ell_1$-norm approximation of optimal sparsity recovery that is known to perform close to optimal for considerably sparse signals, i.e., $L\ll N$. For large choices of $L$, the selection matrix is not significantly sparse anymore, and thus the sparse-recovery-based strategies perform poorly as compared to the \ac{pdd} scheme. It is further observed in both the figures that as the number of selected antennas increases, the gap between different approaches shrinks. This is expected as in the extreme case of $L=N$ all algorithms perform the same. 

\subsubsection{Independent vs Correlated Channel Fading}
In Fig.\ref{fig:fig2}, we investigate the impact of the number of selected antennas on the aggregation error. It is observed that the aggregation error decreases for all the schemes as the number of selected antennas (or equivalently the number of \ac{rf} chains) increases. This is the direct result of the fact that with large active beams at the \ac{ps}, any desired combination of the signals can be readily calculated over the air, and hence the interference among the devices can be avoided. Despite this behavior, the figure indicates that a reasonable aggregation error is achievable by utilizing even less than $50\%$ of the available antennas at the \ac{ps}.  Fig.\ref{fig:fig2} shows the same behavior as the one seen in Fig.~\ref{fig:fig1}. For sake of comparison, we have further plotted the figure for both the standard Rayleigh fading (Fig.~\ref{fig:fig2}(a)) and the correlated Rayleigh fading (Fig.~\ref{fig:fig2}(b)) specified by the correlation model described in Section~\ref{sec:comm_sett}. From the figure, we can see a slight improvement with the idealistic standard Rayleigh model. This difference is however considerably slight, as our proposed schemes compensate for the channel impact via proper beamforming and power scaling. 
\subsubsection{Testing Accuracy} In Fig.~\ref{fig:fig3}, the testing accuracy for the training of the \ac{cnn} over the homogeneous network (i.i.d. data distribution) with FMNIST dataset is plotted against the number of communication rounds. The observed behaviors in Figs.~\ref{fig:fig1} and \ref{fig:fig2} are further seen in this figure: the proposed schemes outperform the baselines. With a small number of active antennas, Lasso and \ac{fista} perform the best, while for larger values of $L$, \ac{pdd} shows the best performance. The results further show that the testing accuracy of all-antenna selection policy can be achieved by the proposed scheme, with even less than $50\%$ of antennas being active. This result indicates that with a proper antenna selection strategy, a desired learning performance can be achieved in the network at a considerably lower implementation cost. 
\subsubsection{Heterogeneous Data Sets} The results for a heterogeneous network (non-i.i.d. data distribution) is further shown in Fig.~\ref{fig:fig4}. The figure reports the same observation with smaller gap among different schemes which is, due to the data heterogeneity. We finally show the results for the CIFAR-10 dataset with both \ac{iid} and non-\ac{iid} distributions in Fig.~\ref{fig:fig5}, considering the case with $L=64$ active antennas. As the figure depicts, unlike the baseline schemes, the proposed schemes show higher robustness against data heterogeneity. This follows the fact that the proposed schemes take both the learning and communication aspects into account.
	\begin{figure*}[t!]
	\centering
	\begin{subfigure}[t]{0.5\textwidth}
		\centering
		\input{cifarL64iid.tex} 
		\caption{i.i.d CIFAR10 dataset}
		\label{fig4a}
	\end{subfigure}%
	~ 
	\begin{subfigure}[t]{0.5\textwidth}
		\centering
		\input{cifar64.tex} 
		\caption{Non-\ac{iid} CIFAR10 dataset}
		\label{fig4b}
	\end{subfigure}
	\caption{Testing accuracy on CIFAR10 dataset for both \ac{iid} and non-\ac{iid} dataset when $L=64$ and SNR=20 dB.}
	\label{fig:fig5}
\end{figure*}
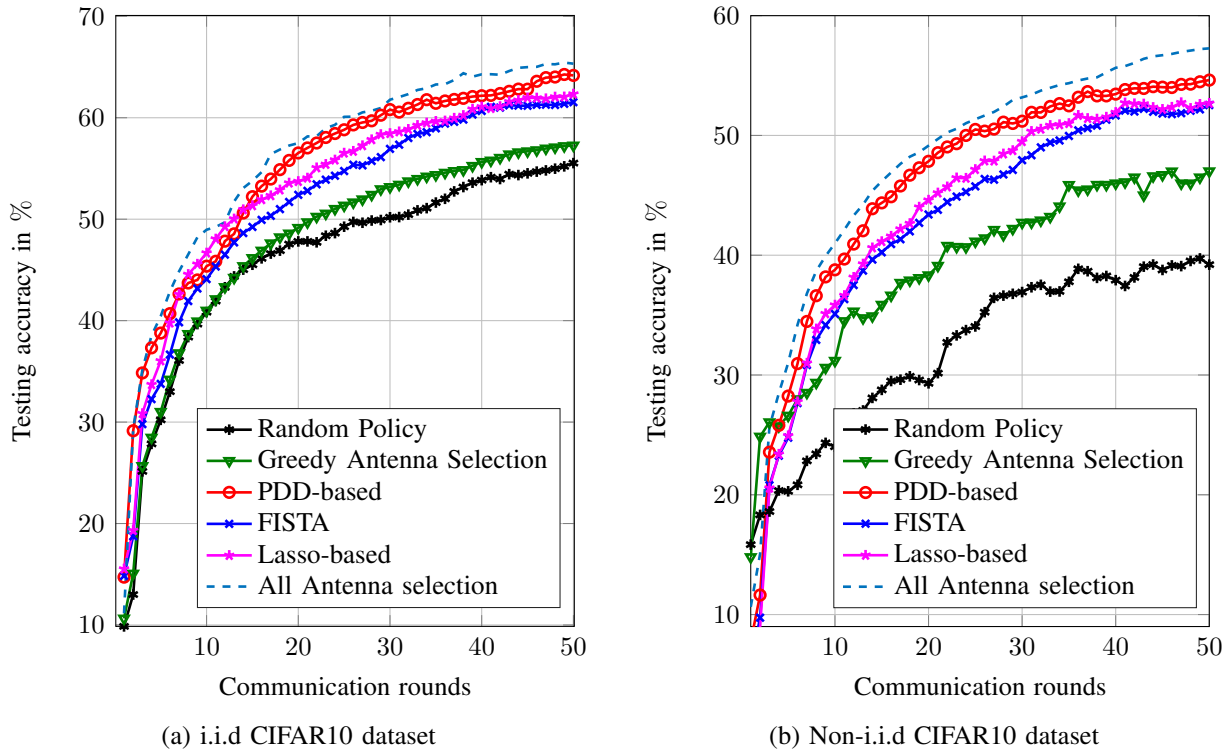
\color{black}
\section{Conclusions}
This paper studied joint communication and aggregation design for \ac{ota-fl} in massive \ac{mimo} systems with reduced \ac{rf} complexity. Our design invoked the \ac{as} scheme to reduce the implementational cost and complexity. We addressed the joint design problem by developing three algorithms based on the \ac{pdd}, Lasso and soft-thresholding methods, catering to a variety of network scenarios with different computational complexity requirements. %
Our investigations demonstrate that the algorithms based on Lasso and soft-thresholding perform closely with a minor performance-complexity trade-off: the Lasso-type algorithm slightly outperforms the soft-thresholding approach, while the latter requires less efforts for tuning and thus has lower complexity. This behavior was expected, since both approaches select antennas via $\ell_1$-norm sparse recovery. The \ac{pdd}-based algorithm however performs differently as compared to the sparse-recovery-based techniques. When the number of active antennas is small, the Lasso-type and soft-thresholding-based algorithms outperform the \ac{pdd}-based approach. Nevertheless, as the number of active antennas increases, the \ac{pdd}-based approach starts to perform superior. This observation follows the fact that both the Lasso-type and soft-thresholding-based algorithms are sparse recovery techniques which perform well at high sparsity, i.e., extremely lower number of active antennas as compared to the \ac{ps} array size. Comparing our results (with only a few active antennas) with the all antenna selection scenario (with all the antennas active), we found that the achievable testing accuracy closely tracks the all antenna selection performance. This indicates the effectiveness of \ac{as}-based massive \ac{mimo} transmission for \ac{ota-fl}.
\appendices
\section{Proof of Theorem~\ref{theorem:P4}}
\label{Proof_th_P4}
We start the proof by considering a selection vector $\bss$ whose $\ell_0$-norm reads $\norm{\bss}_0 < L$. Let $\ell$ be an index out of the support of $\bss$. This means that $s_\ell = 0$. We now construct a new selection vector $\hat{\bss}$ as $\hat{\bss} = \bss + \be_\ell$, where $\be_\ell$ denotes the one-hot vector of length $N$ with the non-zero entry at index $\ell$, i.e., $\rme_\ell = 1$ and $\rme_j = 0$ for $j\neq \ell$. For this selection, the aggregation error reads
\begin{subequations}
\begin{align}
	\epsilon \brc{ \bmm, \hat{\bss}, \mB} &= \Vert {\mathbf m}^\her \Diag{\hat{\bss}}\mH\mB -\boldsymbol{\phi}^\her\Vert^2+\sigma^2 \Vert {\mathbf m}^\her \Diag{\hat{\bss}}\Vert^2\\
	&= \Vert {\mathbf m}^\her \mS \mH\mB -\boldsymbol{\phi}^\her + m_\ell^* \baa^\trp_\ell \Vert^2+\sigma^2 \brc{\Vert {\mathbf m}^\her \mS \Vert^2 + \abs{m_\ell}^2},
\end{align}
\end{subequations}
where $\mS = \Diag{\bss}$ and $\baa_\ell^\trp$ denotes the $\ell$-th row of $\mH\mB$.
We now construct the receiver vector $\bar{\bmm}$ from $\bmm$ by setting its entry $\ell$ to zero, i.e., $\bar{m}_\ell = 0$ and $\bar{m}_n = m_n$ for $n\neq \ell$. Noting that $s_\ell = 0$, we can conclude that $\bmm^\her \mS = \bar{\bmm}^\her \mS$, and thus
\begin{subequations}
\begin{align}
	\epsilon \brc{ \bmm, \hat{\bss}, \mB} &= \Vert \bar{\bmm}^\her \mS \mH\mB -\boldsymbol{\phi}^\her + m_\ell^* \baa^\trp_\ell \Vert^2+\sigma^2 \brc{\Vert \bar{\bmm}^\her \mS \Vert^2 + \abs{m_\ell}^2}\\
	&= \bar{\epsilon} \brc{ \bar{\bmm}, m_\ell, \hat{\bss}, \mB},
\end{align}
\end{subequations}
where we define $\bar{\epsilon} \brc{ \bar{\bmm}, m_\ell, \hat{\bss}, \mB}$ for consistency. From the notation, it is readily seen that
\begin{align}
\epsilon \brc{ \bmm, {\bss}, \mB} = \bar{\epsilon} \brc{ \bar{\bmm}, 0, \hat{\bss}, \mB}.   
\end{align}
For the vector $\hat{\bss}$, the minimum aggregation error subject to the transmit power is given~by
\begin{align}\label{Q_1}
	\minSub{ \min_{m_\ell} \bar{\epsilon} \brc{  \bar{\bmm}, m_\ell, \hat{\bss}, \mB}}
	{\bar{\bmm}, \mB }
	{C: \abs{b_k}^2 < P \text{ for } k \in \dbc{K}. }
	{\maq_1}
\end{align}
The inner minimization is a standard quadratic optimization; therefore, we can conclude that if
	$\norm{\bar{\bmm}^\her \mS \mH\mB -\boldsymbol{\phi}^\her}_\infty \neq 0$
then the minimizer, i.e., $m_\ell^\star = \argmin_{m_\ell} \bar{\epsilon} \brc{  \bar{\bmm}, m_\ell, \hat{\bss}, \mB}$ is non zero. 
This means that
$\min_{m_\ell} \bar{\epsilon} \brc{  \bar{\bmm}, m_\ell, \hat{\bss}, \mB} < \bar{\epsilon} \brc{  \bar{\bmm}, 0 , \hat{\bss}, \mB},$
for a given pair $\brc{\bar{\bmm}, \mB}$. Hence, we have
\begin{align}\label{eq:Non_eq}
	\min_{\bmm,\mB} \epsilon \brc{ \bmm, \hat{\bss}, \mB} \; \text{ s.t. } C < \min_{\bmm,\mB} \epsilon \brc{ \bmm, {\bss}, \mB} \; \text{ s.t. } C
\end{align}
with $C$ referring to the constraint in \ref{Q_1}. The inequality in \eqref{eq:Non_eq} indicates that given the constraint in $	\norm{\bar{\bmm}^\her \mS \mH\mB -\boldsymbol{\phi}^\her}_\infty \neq 0$, by setting an entry of the selection vector to zero the objective function in \ref{P_4} always increases. Hence, the solution of \ref{P_4} satisfies $\norm{\bss}_0 = L$. This concludes the proof.

\bibliographystyle{IEEEtran}
\bibliography{references}
\end{document}

%% file: fig1-16.tex
%
%
\definecolor{mycolor1}{rgb}{0.00000,0.49804,0.00000}%
\definecolor{mycolor2}{rgb}{0.00000,0.44706,0.74118}%
\definecolor{mycolor3}{rgb}{1.00000,0.00000,1.00000}%
\begin{tikzpicture}

\begin{axis}[%
width=2.4in,
height=3.2in,
at={(2.058in,1.005in)},
scale only axis,
scale only axis,
xmin=-9.31533164189925,
xmax=8,
xlabel style={font=\color{white!15!black}},
xlabel={SNR in [dB]},
ymin=-8.68382595682009,
ymax=15.0564587541035,
ylabel style={font=\color{white!15!black}},
ylabel={Aggregation error in [dB]},
axis background/.style={fill=white},
ticklabel style={font=\small},
xtick={-5,0,5},
xticklabels={{$-5$},{$0$},{$5$}},
ytick={-5,0,5,10,15},
yticklabels={{$-5$},{$0$},{$5$},{$10$},{$15$}},
xmajorgrids,
ymajorgrids,
legend style={at={(0.021123,0.029)}, anchor=south west, nodes={scale=0.85, transform shape}, legend cell align=left, align=left, draw=white!15!black}
]
\addplot [color=black, line width=1.0pt, mark=asterisk, mark options={solid, black}]
  table[row sep=crcr]{%
-20	16.2686452595903\\
-16	15.6756280593877\\
-12	15.0476653402766\\
-8	14.5021935990017\\
-4	13.6829027356284\\
0	13.0562018312278\\
4	12.6294531575947\\
8	12.2005160174017\\
12	11.8567490254267\\
16	11.6451507144424\\
20	11.295915311656\\
24	11.3699283839185\\
28	11.363486297236\\
};
\addlegendentry{Random Policy}

\addplot [color=mycolor1, line width=1.0pt, mark=triangle, mark options={solid, rotate=180, mycolor1}]
  table[row sep=crcr]{%
-20	16.0671507290934\\
-16	15.4480694512179\\
-12	14.7841430214441\\
-8	14.1085767401106\\
-4	13.4380133862433\\
0	12.8562521681079\\
4	12.3905049143647\\
8	12.0121453418495\\
12	11.6928017262792\\
16	11.4890418990247\\
20	11.277014619758\\
24	11.1805042756415\\
28	11.1360931884824\\
};
\addlegendentry{Greedy Antenna Selection}

\addplot [color=red, line width=1.0pt, mark=o, mark options={solid, red}]
  table[row sep=crcr]{%
-20	16.1390051442274\\
-16	15.3461717020435\\
-12	14.4604117095416\\
-8	13.4767522630921\\
-4	12.6122623204607\\
0	11.9625018805812\\
4	11.1934756106509\\
8	10.6040804494399\\
12	10.3144570457947\\
16	10.0416329568857\\
20	9.91556556633492\\
24	9.98522480980441\\
28	9.98743279446558\\
};
\addlegendentry{PDD-based}

\addplot [color=blue, line width=1.0pt, mark=x, mark options={solid, blue}]
  table[row sep=crcr]{%
-20	15.6235151196264\\
-16	14.8325851417727\\
-12	13.9441037555043\\
-8	12.9489629784113\\
-4	11.9562332566647\\
0	10.9059075546599\\
4	9.36275188398607\\
8	7.97630541405946\\
12	9.80891855141648\\
16	9.65082407978415\\
20	9.60769748337162\\
24	9.67991998125456\\
28	9.60842948903235\\
};
\addlegendentry{FISTA}

\addplot [color=mycolor3, line width=1.0pt, mark=star, mark options={solid, mycolor3}]
  table[row sep=crcr]{%
-20	15.6486575693487\\
-16	14.802561377007\\
-12	13.8399194335358\\
-8	12.9010783094541\\
-4	11.9209700379471\\
0	10.8437936755453\\
4	9.0751861989248\\
8	7.18829038884253\\
12	9.88857951236384\\
16	9.57565816680643\\
20	9.51060216267358\\
24	9.53301363264503\\
28	9.4974793839868\\
};
\addlegendentry{Lasso-based}

\addplot [color=mycolor2, dashed, line width=1.0pt]
  table[row sep=crcr]{%
-20	13.5950158793215\\
-16	11.5910619540777\\
-12	8.93074511051065\\
-8	5.70574991537345\\
-4	2.09163290162314\\
0	-1.73388088883534\\
4	-5.65463453155225\\
8	-9.62454371353991\\
12	-13.6125866670118\\
16	-17.6078088669363\\
20	-21.6059039800101\\
24	-25.6051451842418\\
28	-29.6048430313447\\
};
\addlegendentry{All Antenna Selection}

\end{axis}
\end{tikzpicture}%

%% file: fig1-64.tex
%
%
\definecolor{mycolor1}{rgb}{0.00000,0.49804,0.00000}%
\definecolor{mycolor2}{rgb}{1.00000,0.00000,1.00000}%
\definecolor{mycolor3}{rgb}{0.00000,0.44706,0.74118}%
\begin{tikzpicture}

\begin{axis}[%
width=2.4in,
height=3.2in,
at={(2.058in,1.005in)},
scale only axis,
xmin=-9,
xmax=8,
xlabel style={font=\color{white!15!black}},
xlabel={SNR in [dB]},
ymin=-10.2204074064293,
ymax=11.3305512954202,
ylabel style={font=\color{white!15!black}},
ylabel={Aggregation error in [dB]},
axis background/.style={fill=white},
ticklabel style={font=\small},
xtick={-5,0,5},
xticklabels={{$-5$},{$0$},{$5$}},
ytick={-10,-5,0,5,10,15},
yticklabels={{-10},{$-5$},{$0$},{$5$},{$10$}},
xmajorgrids,
ymajorgrids,
legend style={at={(0.81,0.349)},legend cell align=left, align=left, nodes={scale=0.85, transform shape}, draw=white!15!black}
]

\addplot [color=black, line width=1.0pt, mark=asterisk, mark options={solid, black}]
  table[row sep=crcr]{%
-20	15.9255701043535\\
-16	14.4991924087051\\
-12	12.6655553720714\\
-8	9.94244573115605\\
-4	7.21307903881051\\
0	4.08505885008687\\
4	1.55280623495256\\
8	-2.33573442604917\\
12	-5.43907375701483\\
16	-8.84987329140498\\
20	-12.5963635154745\\
24	-14.6787923270706\\
28	-17.6505075611976\\
};
\addlegendentry{Random Policy}

\addplot [color=mycolor1, line width=1.0pt, mark=triangle, mark options={solid, rotate=180, mycolor1}]
  table[row sep=crcr]{%
-20	15.194463652938\\
-16	14.2465516487638\\
-12	12.0161597170814\\
-8	9.5040702362518\\
-4	6.60678022417995\\
0	3.40350530272021\\
4	0.796988391697428\\
8	-2.95337224122477\\
12	-5.99550262493235\\
16	-9.50809325804602\\
20	-12.83730470338\\
24	-15.6652466096161\\
28	-17.6205715720225\\
};
\addlegendentry{Greedy Antenna Selection}
\addplot [color=red, line width=1.0pt, mark=o, mark options={solid, red}]
  table[row sep=crcr]{%
-20	13.3\\
-16	12.6\\
-12	10.3\\
-8	7.4\\
-4	3.6\\
0	-0.32\\
4	-4.2\\
8	-8.2\\
12	-11.99\\
16	-15.6\\
20	-18\\
24	-20\\
28	-22\\
};
\addlegendentry{PDD-based}
\addplot [color=blue, line width=1.0pt, mark=x, mark options={solid, blue}]
  table[row sep=crcr]{%
-20	13.8388406323304\\
-16	12.8947804512759\\
-12	11.2\\
-8	8.46393451349503\\
-4	4.61699052902447\\
0	0.511169905701149\\
4	-2.94574841108092\\
8	-6.49023373268899\\
12	-9.90345191581243\\
16	-13.5051426454119\\
20	-16.7278226806913\\
24	-19.6641068886379\\
28	-22.1352747073269\\
};
\addlegendentry{FISTA}
\addplot [color=mycolor2, line width=1.0pt, mark=star, mark options={solid, mycolor2}]
  table[row sep=crcr]{%
-20	13.5\\
-16	12.6\\
-12	11\\
-8	8.1\\
-4	4.3\\
0	0.2\\
4	-3.2\\
8	-6.7\\
12	-10.12\\
16	-13.7\\
20	-17\\
24	-20.3\\
28	-22.5\\
};
\addlegendentry{Lasso-based}

\addplot [color=mycolor3, dashed, line width=1.0pt]
  table[row sep=crcr]{%
-20	13.348991488496\\
-16	11.3466594051332\\
-12	8.65182446996102\\
-8	5.39043739427935\\
-4	1.75305773989357\\
0	-2.08364498787673\\
4	-6.00925850765309\\
8	-9.96685939037516\\
12	-13.9492580070072\\
16	-17.9447148523609\\
20	-21.9429035752786\\
24	-25.9421820773973\\
28	-29.9418947779977\\
};
\addlegendentry{All Antenna Selection}

\end{axis}
\end{tikzpicture}%

%% file: MSELFinal.tex
%
%
\definecolor{mycolor1}{rgb}{1.00000,0.00000,1.00000}%
\definecolor{mycolor2}{rgb}{0.00000,0.49804,0.00000}%
\begin{tikzpicture}

\begin{axis}[%
width=2.6in,
height=3.in,
at={(1.219in,0.764in)},
scale only axis,
xmin=1,
xmax=126,
ymin=-2,
ymax=16,
ylabel={Aggregation error in [dB]},
xlabel={Number of selected antennas $L$},
xlabel style={ font=\small},
ylabel style={ font=\small},
axis background/.style={fill=white},
ticklabel style={font=\small},
xtick={20,40,60,80,100,120},
xticklabels={{$20$},{$40$},{$60$},{$80$},{$100$},{$120$}},
ytick={-2,0,2,4,6,8,10,12,14,16},
yticklabels={{$-2$},{$0$},{$2$},{$4$},{$6$},{$8$},{$10$},{$12$},{$14$},{$16$}},
xmajorgrids,
ymajorgrids,
legend style={legend cell align=left, align=left, draw=white!15!black, nodes={scale=0.85, transform shape},}
]

\addplot [color=black, line width=1.0pt, mark=asterisk, mark options={solid, black}]
  table[row sep=crcr]{%
1	15.6101296173976\\
6	14.5666894827431\\
11	13.9114995317098\\
16	13.2522309315854\\
21	12.2876309491646\\
26	11.6671702602869\\
31	10.9608728431935\\
36	9.71381014760948\\
41	8.98946675777798\\
46	8.84702393975841\\
51	7.23345609584496\\
56	6.57213647085342\\
61	5.97424707129785\\
66	4.5206310954023\\
71	4.04970939463428\\
76	2.93216497317635\\
81	2.35944746485123\\
86	1.98554133572506\\
91	1.65225833009467\\
96	0.853485141005494\\
101	0.669898473203709\\
106	0.282063359691967\\
111	-0.123706080399392\\
116	-0.625692117434753\\
121	-0.77342108026422\\
126	-1.26073998740007\\
};
\addlegendentry{Random Policy}

\addplot [color=mycolor2, line width=1.0pt, mark=triangle, mark options={solid, rotate=180, mycolor2}]
table[row sep=crcr]{%
	1	15.6040567196859\\
	6	14.4350045304392\\
	11	13.7651945770187\\
	16	13.0582128523916\\
	21	12.3410050191042\\
	26	11.5698325453317\\
	31	10.2629721090301\\
	36	9.57908987123179\\
	41	9.15494049681765\\
	46	8.70827367558833\\
	51	6.78963363752697\\
	56	6.04390287497765\\
	61	5.14095711389218\\
	66	3.87463613919227\\
	71	3.40096968901302\\
	76	2.49972568884251\\
	81	2.0542910092749\\
	86	1.82449335014187\\
	91	1.61684885956902\\
	96	0.749101645218855\\
	101	0.291996376076192\\
	106	0.00978156976453046\\
	111	-0.395590898230789\\
	116	-0.70309590463112\\
	121	-0.803765056662598\\
	126	-1.27892638244211\\
};
\addlegendentry{Greedy Antenna Selection}

\addplot [color=red, line width=1.0pt, mark=o, mark options={solid, red}]
table[row sep=crcr]{%
	1	15.1638972356548\\
	6	13.7239611888857\\
	11	12.9902970401536\\
	16	11.905501236719771\\
	21	11.0112656803708\\
	26	9.721186710054915\\
	31	7.83887469051903\\
	36	5.90766231949668\\
	41	4.20447796651607\\
	46	2.90317297067132\\
	51	1.5129173758161\\
	56	0.579024727948824\\
	61	0.134571001352477\\
	66	-0.2422939715501\\
	71	-0.5146474517314\\
	76	-0.80718815080652\\
	81	-0.94861833608468\\
	86	-0.999999536608965\\
	91	-1.1045550101708\\
	96	-1.17831749375976\\
	101	-1.18758457766424\\
	106	-1.17135528856294\\
	111	-1.17982282592153\\
	116	-1.1976670574571\\
	121	-1.22687686105718\\
	126	-1.25034797672307\\
};
\addlegendentry{PDD-based}

\addplot [color=blue, line width=1.0pt, mark=x, mark options={solid, blue}]
  table[row sep=crcr]{%
1	15.1634094484235\\
6	13.6642112541947\\
11	12.3044292472666\\
16	11.3226955101012\\
21	9.77908467292789\\
26	8.70810003025715\\
31	7.17387126296715\\
36	5.938838540364419\\
41	4.73489412344423\\
46	4.035615606652761\\
51	3.18264657726911\\
56	2.52485998750323\\
61	1.59908290554805188\\
66	0.94240369558736\\
71	0.650995424547816\\
76	0.3547115078587\\
81	0.211614451781378\\
86	-0.0224882071777025\\
91	-0.0487270213907\\
96	-0.29021944580447\\
101	-0.53052461292211\\
106	-0.631794039077333\\
111	-0.713573811553145\\
116	-1.004322369399\\
121	-0.923516664845202\\
126	-1.29219928201601\\
};
\addlegendentry{FISTA}

\addplot [color=mycolor1, line width=1.0pt, mark=star, mark options={solid, mycolor1}]
table[row sep=crcr]{%
	1	15.1629216611921\\
	6	13.6044613195037\\
	11	12.1185614543795\\
	16	11.0952673482253\\
	21	9.4455125421478\\
	26	8.20433295996514\\
	31	6.50886783541527\\
	36	4.96910848779169\\
	41	4.22500858172773\\
	46	3.48058242634201\\
	51	2.57355557292815\\
	56	1.72177115109046\\
	61	1.23056337929984\\
	66	0.56230554102361625\\
	71	0.230526103922493\\
	76	-0.1340141830798051\\
	81	-0.2488090737171223\\
	86	-0.3595300752519514\\
	91	-0.4532955059861562\\
	96	-0.74513104843666\\
	101	-0.864562689008237\\
	106	-0.999116934449061\\
	111	-1.02891936389102\\
	116	-1.24193899308759\\
	121	-1.07826471911861\\
	126	-1.33405058730895\\
};
\addlegendentry{Lasso-based}

\end{axis}
\end{tikzpicture}%

%% file: iidL.tex
%
%
\definecolor{mycolor1}{rgb}{0.00000,0.49804,0.00000}%
\definecolor{mycolor2}{rgb}{1.00000,0.00000,1.00000}%
\begin{tikzpicture}

\begin{axis}[%
width=2.6in,
height=3.in,
at={(1.121in,0.658in)},
scale only axis,
xmin=1,
xmax=126,
ymin=-4,
ymax=16,
ylabel={Aggregation error in [dB]},
xlabel={Number of selected antennas $L$},
xlabel style={ font=\small},
ylabel style={ font=\small},
axis background/.style={fill=white},
ticklabel style={font=\small},
xtick={20,40,60,80,100,120},
xticklabels={{$20$},{$40$},{$60$},{$80$},{$100$},{$120$}},
ytick={-2,0,2,4,6,8,10,12,14,16},
yticklabels={{$-2$},{$0$},{$2$},{$4$},{$6$},{$8$},{$10$},{$12$},{$14$},{$16$}},
axis background/.style={fill=white},
xmajorgrids,
ymajorgrids,
legend style={legend cell align=left, align=left, draw=white!15!black,nodes={scale=0.85, transform shape}}
]

\addplot [color=black, line width=1.0pt, mark=asterisk, mark options={solid, black}]
table[row sep=crcr]{%
	1	15.6380155208293\\
	6	14.6051895354226\\
	11	13.5941520879762\\
	16	13.1470646475579\\
	21	12.228357533145\\
	26	11.3877303314715\\
	31	10.7284827577803\\
	36	9.81590652438988\\
	41	9.09336809350693\\
	46	7.0835555467745\\
	51	7.07005194302519\\
	56	6.24934051468128\\
	61	4.59672675175464\\
	66	3.79538308466509\\
	71	3.23050421724387\\
	76	2.03757837132125\\
	81	1.71075518972933\\
	86	0.826283715283597\\
	91	0.457602094918414\\
	96	0.302189893944733\\
	101	-0.576632498548045\\
	106	-0.475674275030714\\
	111	-1.20031499245511\\
	116	-1.13189423621403\\
	121	-1.90305126104728\\
	126	-2.3\\
};
\addlegendentry{Random Policy}

\addplot [color=mycolor1, line width=1.0pt, mark=triangle, mark options={solid, rotate=180, mycolor1}]
table[row sep=crcr]{%
	1	15.57643555514\\
	6	14.5665290885131\\
	11	13.809450862758\\
	16	12.7472193227652\\
	21	12.0002190725187\\
	26	11.1595601063278\\
	31	10.5161703268562\\
	36	9.80688675243495\\
	41	8.42280498275652\\
	46	7.52537530229734\\
	51	6.97221991309639\\
	56	5.15853090928826\\
	61	4.07167816619344\\
	66	3.12942727607452\\
	71	3.2883351912591\\
	76	1.87246382322683\\
	81	1.06530028717215\\
	86	0.347851217451049\\
	91	0.370586493481399\\
	96	0.209719331228165\\
	101	-0.75356118170827\\
	106	-0.653802638488973\\
	111	-1.41356773919709\\
	116	-1.28833277730601\\
	121	-2.04035924127221\\
	126	-2.3\\
};
\addlegendentry{Greedy Antenna Selection}

\addplot [color=red, line width=1.0pt, mark=o, mark options={solid, red}]
table[row sep=crcr]{%
	1	15.2017527511837\\
	6	13.9167279081427\\
	11	12.8405897302359\\
	16	11.836290139193\\
	21	10.7018225161824\\
	26	8.95086328638562\\
	31	7.4789143079304\\
	36	5.4287651251485\\
	41	4.40124478055238\\
	46	2.98647089466478\\
	51	1.82724226351035\\
	56	0.726973806290375\\
	61	0.182371116663138\\
	66	-0.402849201126248\\
	71	-0.68558099529678\\
	76	-1.01140443497106\\
	81	-1.4053073030139\\
	86	-1.85121911491359\\
	91	-1.9229184342946\\
	96	-1.9834621568545\\
	101	-2.03500420465798\\
	106	-2.16118708573457\\
	111	-2.20756955667167\\
	116	-2.3081347226747\\
	121	-2.324953094351\\
	126	-2.35\\
};
\addlegendentry{PDD-based}

\addplot [color=blue, line width=1.0pt, mark=x, mark options={solid, blue}]
  table[row sep=crcr]{%
1	15.2685067929582\\
6	13.631550356683\\
11	12.2974132108034\\
16	10.6318018390405\\
21	9.55771439258888\\
26	7.8366142057257\\
31	6.77127685732742\\
36	4.93613313335976\\
41	3.84898368879344\\
46	2.9395005070641\\
51	2.24319393371472\\
56	1.72873197943875\\
61	0.995383873722551\\
66	0.876083177123988\\
71	0.59837507995955\\
76	0.139242037954462\\
81	-0.0928166790583\\
86	-0.493220844913458\\
91	-1.09787205020212\\
96	-1.19948542744267\\
101	-1.3200936786139\\
106	-1.51313118812112\\
111	-1.6971364823217\\
116	-1.89315204861135\\
121	-2.02512024341996\\
126	-2.3\\
};
\addlegendentry{FISTA}

\addplot [color=mycolor2, line width=1.0pt, mark=star, mark options={solid, mycolor2}]
  table[row sep=crcr]{%
1	15.2907581402164\\
6	13.5698245061964\\
11	12.2163543709926\\
16	10.463639072323\\
21	9.45355843957909\\
26	7.59853117883906\\
31	6.53539770712643\\
36	4.77192246943018\\
41	3.66489665820712\\
46	2.95124310396427\\
51	2.13920601616363\\
56	1.47829243615166\\
61	0.792130684457698\\
66	0.556350082561429\\
71	0.277386061145467\\
76	-0.148419580276918\\
81	-0.4209393350472\\
86	-0.832720412413491\\
91	-1.30413364622524\\
96	-1.39547960979563\\
101	-1.49882131012492\\
106	-1.67514516252449\\
111	-1.82474475090919\\
116	-1.99689771712719\\
121	-2.10007845615272\\
126	-2.3125\\
};
\addlegendentry{Lasso-based}

\end{axis}
\end{tikzpicture}%

%% file: figmnistL16.tex
%
%
\definecolor{mycolor1}{rgb}{0.00000,0.44706,0.74118}%
\definecolor{mycolor2}{rgb}{1.00000,0.00000,1.00000}%
\definecolor{mycolor3}{rgb}{0.00000,0.49804,0.00000}%
\begin{tikzpicture}

\begin{axis}[%
width=2.4in,
height=3.2in,
at={(1.843in,0.738in)},
scale only axis,
separate axis lines,
xmin=1,
xmax=50,
xlabel={Communication rounds},
ymin=65,
ymax=90,
ylabel={Testing accuracy in $\%$},
xlabel style={ font=\small},
ylabel style={ font=\small},
axis background/.style={fill=white},
ticklabel style={font=\small},
xtick={0,10,20,30,40,50},
xticklabels={{$0$},{$10$},{$20$},{$30$},{$40$},{$50$}},
ytick={55,60,65,70,75,80,85,90},
yticklabels={{$55$},{$60$},{$65$},{$70$},{$75$},{$80$},{$85$},{$90$}},
xmajorgrids,
ymajorgrids,
legend style={at={(0.97,0.03)},  nodes={scale=0.85, transform shape}, anchor=south east, legend cell align=left, align=left, draw=white!15!black}
]
\addplot [color=black, line width=1.0pt, mark=asterisk, mark options={solid, black}]
  table[row sep=crcr]{%
1	46.245\\
2	45.665\\
3	50.595\\
4	61.93\\
5	68.885\\
6	70.93\\
7	71.86\\
8	72.74\\
9	73.095\\
10	73.385\\
11	74.105\\
12	74.36\\
13	74.385\\
14	74.775\\
15	75.09\\
16	75.085\\
17	75.56\\
18	75.65\\
19	75.72\\
20	75.15\\
21	76.26\\
22	76.975\\
23	77.005\\
24	77.345\\
25	77.33\\
26	77.65\\
27	77.825\\
28	77.775\\
29	77.885\\
30	78.14\\
31	78.225\\
32	78.2795\\
33	78.251\\
34	78.34\\
35	78.355\\
36	78.38\\
37	78.428\\
38	78.4291\\
39	78.4195\\
40	78.45\\
41	78.46045\\
42	78.48\\
43	78.577\\
44	78.574\\
45	78.665\\
46	78.6825\\
47	78.75\\
48	78.8415\\
49	78.941\\
50	79.000535\\
};
\addlegendentry{Random Policy}

\addplot [color=mycolor3, line width=1.0pt, mark=triangle, mark options={solid, rotate=180, mycolor3}]
  table[row sep=crcr]{%
1	44.98\\
2	43.96\\
3	54.83\\
4	66.7\\
5	70.96\\
6	71.67\\
7	72.25\\
8	72.98\\
9	73.14\\
10	73.91\\
11	73.923\\
12	74.501\\
13	75.134\\
14	75.24\\
15	75.46\\
16	75.6278\\
17	75.93141\\
18	76.07404\\
19	76.1855\\
20	76.3882\\
21	76.489\\
22	77.4596023\\
23	77.576504\\
24	77.693634\\
25	77.846803\\
26	77.937038\\
27	78.139756\\
28	78.348066\\
29	78.5863\\
30	78.89066\\
31	78.9402\\
32	79.185\\
33	79.3\\
34	79.299\\
35	79.58\\
36	79.6139\\
37	79.8\\
38	79.9204\\
39	80.1214\\
40	80.2543\\
41	80.29\\
42	80.355\\
43	80.395\\
44	80.4154\\
45	80.5282\\
46	80.539\\
47	80.45\\
48	80.53\\
49	80.605\\
50	80.731\\
};
\addlegendentry{Greedy Antenna Selection}

\addplot [color=red, line width=1.0pt, mark=o, mark options={solid, red}]
  table[row sep=crcr]{%
1	56.3466666666667\\
2	58.4133333333333\\
3	64.5416666666667\\
4	71.0766666666667\\
5	73.87\\
6	74.9083333333333\\
7	75.5716666666667\\
8	76.2383333333333\\
9	76.4183333333333\\
10	77.0633333333333\\
11	77.3681666666667\\
12	77.9388333333333\\
13	78.4836666666667\\
14	78.6883333333333\\
15	79.045\\
16	79.2855666666667\\
17	79.5123716666667\\
18	79.7536866666667\\
19	79.9860833333333\\
20	80.1507666666667\\
21	80.3328333333333\\
22	80.7464678166667\\
23	81.053252\\
24	81.2084836666667\\
25	81.4167348333333\\
26	81.488519\\
27	81.6632113333333\\
28	81.7973663333333\\
29	81.9298166666667\\
30	82.10533\\
31	82.2951\\
32	82.4541666666667\\
33	82.512\\
34	82.5661666666667\\
35	82.713\\
36	82.72345\\
37	82.8615\\
38	82.9145333333333\\
39	83.0208666666667\\
40	83.0788166666667\\
41	83.1276666666667\\
42	83.1611666666667\\
43	83.1665\\
44	83.1690333333333\\
45	83.2156\\
46	83.2631666666667\\
47	83.2083333333333\\
48	83.2541666666667\\
49	83.3226666666667\\
50	83.4171666666667\\
};
\addlegendentry{PDD-based}

\addplot [color=blue, line width=1.0pt, mark=x, mark options={solid, blue}]
  table[row sep=crcr]{%
1	66.18\\
2	69.93\\
3	70.81\\
4	72.72\\
5	74.34\\
6	75.57\\
7	76.47\\
8	77.1\\
9	77.33\\
10	77.93\\
11	78.74\\
12	79.41\\
13	79.71\\
14	80.02\\
15	80.67\\
16	81.04\\
17	81.1\\
18	81.34\\
19	82.1\\
20	82.06\\
21	82.46\\
22	82.64\\
23	82.84\\
24	83.04\\
25	83.46\\
26	83.47\\
27	83.67\\
28	83.78\\
29	83.86\\
30	83.93\\
31	83.97\\
32	83.993\\
33	84.01\\
34	84.0409\\
35	84.0701\\
36	84.0923\\
37	84.128\\
38	84.1523\\
39	84.16\\
40	84.194\\
41	84.2036\\
42	84.2174\\
43	84.212\\
44	84.2688\\
45	84.242\\
46	84.209\\
47	84.2584\\
48	84.2699\\
49	84.2716\\
50	84.2945\\
};
\addlegendentry{FISTA}
\addplot [color=mycolor2, line width=1.0pt, mark=star, mark options={solid, mycolor2}]
table[row sep=crcr]{%
	1	66.7133333333333\\
	2	71.8666666666667\\
	3	73.2533333333333\\
	4	74.4533333333333\\
	5	75.78\\
	6	77.1466666666667\\
	7	77.8933333333333\\
	8	78.4966666666667\\
	9	78.6966666666667\\
	10	79.2166666666667\\
	11	79.8133333333333\\
	12	80.3766666666667\\
	13	80.8333333333333\\
	14	81.1366666666667\\
	15	81.63\\
	16	81.9433333333333\\
	17	82.0933333333333\\
	18	82.4333333333333\\
	19	82.7866666666667\\
	20	82.9133333333333\\
	21	83.1766666666667\\
	22	83.4333333333333\\
	23	83.53\\
	24	83.7233333333333\\
	25	83.9866666666667\\
	26	84.04\\
	27	84.1866666666667\\
	28	84.2466666666667\\
	29	84.2733333333333\\
	30	84.32\\
	31	84.65\\
	32	84.7233333333333\\
	33	84.724\\
	34	84.8333333333333\\
	35	84.846\\
	36	84.833\\
	37	84.923\\
	38	84.9086666666667\\
	39	84.9203333333333\\
	40	84.9033333333333\\
	41	84.9653333333333\\
	42	84.9673333333333\\
	43	84.938\\
	44	84.9226666666667\\
	45	84.903\\
	46	84.9873333333333\\
	47	84.9666666666667\\
	48	84.9783333333333\\
	49	85.0403333333333\\
	50	85.1033333333333\\
};
\addlegendentry{Lasso-based}

\addplot [color=mycolor1, dashed, line width=1.0pt]
table[row sep=crcr]{%
	1	67.06\\
	2	74.2366666666667\\
	3	75.81\\
	4	76.8833333333333\\
	5	77.8666666666667\\
	6	78.88\\
	7	79.6933333333333\\
	8	80.4066666666667\\
	9	81.0833333333333\\
	10	81.6533333333333\\
	11	82.1133333333333\\
	12	82.3733333333333\\
	13	82.6166666666667\\
	14	82.9933333333333\\
	15	83.2266666666667\\
	16	83.5166666666667\\
	17	83.73\\
	18	83.9766666666667\\
	19	84.2333333333333\\
	20	84.44\\
	21	84.69\\
	22	84.78\\
	23	84.9533333333333\\
	24	85.17\\
	25	85.26\\
	26	85.4233333333333\\
	27	85.5866666666667\\
	28	85.69\\
	29	85.79\\
	30	85.94\\
	31	85.93\\
	32	86.0866666666667\\
	33	86.0833333333333\\
	34	86.24\\
	35	86.3466666666667\\
	36	86.3366666666667\\
	37	86.3933333333333\\
	38	86.41\\
	39	86.58\\
	40	86.62\\
	41	86.6966666666667\\
	42	86.7133333333333\\
	43	86.8266666666667\\
	44	86.97\\
	45	87.0433333333333\\
	46	87.0366666666667\\
	47	87.1133333333333\\
	48	87.1733333333333\\
	49	87.24\\
	50	87.3533333333333\\
};
\addlegendentry{All Antenna selection}
\end{axis}

\end{tikzpicture}%

%% file: figmnistL40.tex
%
%
\definecolor{mycolor1}{rgb}{0.00000,0.44700,0.74100}%
\definecolor{mycolor2}{rgb}{0.00000,0.49804,0.00000}%
\definecolor{mycolor3}{rgb}{1.00000,0.00000,1.00000}%
\begin{tikzpicture}

\begin{axis}[%
width=2.4in,
height=3.2in,
at={(2.6in,0.718in)},
scale only axis, 
xmin=1,
xmax=50,
ylabel={Testing accuracy in $\%$},
ymin=65,
ymax=90,
ylabel style={font=\color{white!15!black}},
xlabel={Communication rounds},
xlabel style={ font=\small},
ylabel style={ font=\small},
axis background/.style={fill=white},
ticklabel style={font=\small},
xtick={0,10,20,30,40,50},
xticklabels={{$0$},{$10$},{$20$},{$30$},{$40$},{$50$}},
ytick={55,60,65,70,75,80,85,90},
yticklabels={{$55$},{$60$},{$65$},{$70$},{$75$},{$80$},{$85$},{$90$}},
xmajorgrids,
ymajorgrids, 
legend style={at={(0.17,0.03)}, nodes={scale=0.85, transform shape}, anchor=south west, legend cell align=left, align=left, draw=white!15!black}
]

\addplot [color=black, line width=1.0pt, mark=asterisk, mark options={solid, black}]
table[row sep=crcr]{%
	1	66.4066666666667\\
	2	71.8466666666667\\
	3	73.3\\
	4	74.3033333333333\\
	5	75.58\\
	6	76.6533333333333\\
	7	77.1633333333333\\
	8	77.87\\
	9	78.2233333333333\\
	10	78.65\\
	11	79.2133333333333\\
	12	79.6633333333333\\
	13	80.1533333333333\\
	14	80.41\\
	15	80.91\\
	16	81.1566666666667\\
	17	81.6766666666667\\
	18	81.79\\
	19	81.9433333333333\\
	20	82.2933333333333\\
	21	82.2433333333333\\
	22	82.4533333333333\\
	23	82.3866666666667\\
	24	82.3566666666667\\
	25	82.4633333333333\\
	26	82.5133333333333\\
	27	82.54\\
	28	82.5366666666667\\
	29	82.61\\
	30	82.6033333333333\\
	31	82.8266666666667\\
	32	82.8\\
	33	82.8733333333333\\
	34	82.83\\
	35	82.76\\
	36	82.7966666666667\\
	37	82.6966666666667\\
	38	82.7666666666667\\
	39	82.76\\
	40	82.73\\
	41	82.7323333333333\\
	42	82.7953333333333\\
	43	82.8032666666667\\
	44	82.8266666666667\\
	45	82.8590286666667\\
	46	82.8892006666667\\
	47	82.9055666666667\\
	48	82.9333333333334\\
	49	82.9166666666667\\
	50	82.8436666666667\\
};
\addlegendentry{Random Policy}
\addplot [color=mycolor2, line width=1.0pt, mark=triangle, mark options={solid, rotate=180, mycolor2}]
  table[row sep=crcr]{%
1	66.3666666666667\\
2	71.98\\
3	73.1933333333333\\
4	74.26\\
5	75.52\\
6	76.5033333333333\\
7	77.2833333333333\\
8	78.0366666666667\\
9	78.46\\
10	79.0066666666667\\
11	79.3033333333333\\
12	79.69\\
13	80.2533333333333\\
14	80.5666666666667\\
15	80.7266666666667\\
16	81.1866666666667\\
17	81.5466666666667\\
18	82.1066666666667\\
19	82.1133333333333\\
20	82.46\\
21	82.47\\
22	82.6233333333333\\
23	82.7833333333333\\
24	82.6866666666667\\
25	82.7366666666667\\
26	82.76\\
27	82.8066666666667\\
28	82.82\\
29	82.8866666666667\\
30	82.95\\
31	83.0533333333333\\
32	83.22\\
33	83.37\\
34	83.3633333333333\\
35	83.36\\
36	83.5833333333333\\
37	83.6133333333333\\
38	83.7633333333333\\
39	83.75\\
40	83.9333333333333\\
41	83.86\\
42	83.88\\
43	83.72\\
44	83.8666666666667\\
45	83.6233333333333\\
46	83.5233333333333\\
47	83.6533333333333\\
48	83.5733333333333\\
49	83.63\\
50	83.6233333333333\\
};
\addlegendentry{Greedy Antenna Selection}

\addplot [color=red, line width=1.0pt, mark=o, mark options={solid, red}]
table[row sep=crcr]{%
	1	67.22\\
	2	73.6766666666667\\
	3	75.0733333333333\\
	4	76.18\\
	5	77.0633333333333\\
	6	77.94\\
	7	78.62\\
	8	79.3733333333333\\
	9	80.0133333333333\\
	10	80.4266666666667\\
	11	81.1333333333333\\
	12	81.5766666666667\\
	13	81.9833333333333\\
	14	82.3933333333333\\
	15	82.7833333333333\\
	16	83.1833333333333\\
	17	83.46\\
	18	83.71\\
	19	84.01\\
	20	84.1066666666667\\
	21	84.27\\
	22	84.38\\
	23	84.45\\
	24	84.65\\
	25	84.88\\
	26	85.0566666666667\\
	27	85.1633333333333\\
	28	85.2033333333333\\
	29	85.4\\
	30	85.5066666666667\\
	31	85.69\\
	32	85.76\\
	33	85.85\\
	34	85.9966666666667\\
	35	85.9466666666667\\
	36	86.0166666666667\\
	37	86.1733333333333\\
	38	86.13\\
	39	86.14\\
	40	86.2733333333333\\
	41	86.3066666666667\\
	42	86.3733333333333\\
	43	86.3166666666667\\
	44	86.3533333333333\\
	45	86.48\\
	46	86.44\\
	47	86.39\\
	48	86.41\\
	49	86.5033333333333\\
	50	86.3333333333333\\
};
\addlegendentry{PDD-based}

\addplot [color=blue, line width=1.0pt, mark=x, mark options={solid, blue}]
table[row sep=crcr]{%
	1	67.98\\
	2	73.3966666666667\\
	3	74.8333333333333\\
	4	75.93\\
	5	76.8366666666667\\
	6	77.5933333333333\\
	7	78.2733333333333\\
	8	78.9533333333333\\
	9	79.3\\
	10	79.9666666666667\\
	11	80.3833333333333\\
	12	80.8233333333333\\
	13	81.24\\
	14	81.62\\
	15	82.0633333333333\\
	16	82.41\\
	17	82.6766666666667\\
	18	82.9566666666667\\
	19	83.0466666666667\\
	20	83.2766666666667\\
	21	83.5966666666667\\
	22	83.7333333333333\\
	23	83.9766666666667\\
	24	84.2766666666667\\
	25	84.4433333333333\\
	26	84.47\\
	27	84.7\\
	28	84.7766666666667\\
	29	84.8933333333333\\
	30	84.9633333333333\\
	31	85.14\\
	32	85.2233333333333\\
	33	85.2966666666667\\
	34	85.35\\
	35	85.3633333333333\\
	36	85.4866666666667\\
	37	85.4766666666667\\
	38	85.62\\
	39	85.68\\
	40	85.6166666666667\\
	41	85.5966666666667\\
	42	85.6033333333333\\
	43	85.6766666666667\\
	44	85.45\\
	45	85.6266666666667\\
	46	85.5933333333333\\
	47	85.57\\
	48	85.7733333333333\\
	49	85.71\\
	50	85.7166666666667\\
};
\addlegendentry{FISTA}

\addplot [color=mycolor3, line width=1.0pt, mark=star, mark options={solid, mycolor3}]
  table[row sep=crcr]{%
1	66.02\\
2	73.2666666666667\\
3	74.9733333333333\\
4	76.1866666666667\\
5	77.2033333333333\\
6	78.1\\
7	78.7933333333333\\
8	79.5633333333333\\
9	80.15\\
10	80.7066666666667\\
11	81.14\\
12	81.4966666666667\\
13	81.98\\
14	82.3133333333333\\
15	82.6733333333333\\
16	83.15\\
17	83.41\\
18	83.7266666666667\\
19	83.7166666666667\\
20	83.8833333333333\\
21	84.0933333333333\\
22	84.3066666666667\\
23	84.32\\
24	84.59\\
25	84.8766666666667\\
26	84.9366666666667\\
27	84.8933333333333\\
28	85\\
29	85.0666666666667\\
30	85.1166666666667\\
31	85.2866666666667\\
32	85.3666666666667\\
33	85.3466666666667\\
34	85.3333333333333\\
35	85.38\\
36	85.51\\
37	85.48\\
38	85.5366666666667\\
39	85.6266666666667\\
40	85.5866666666667\\
41	85.5233333333333\\
42	85.7633333333333\\
43	85.65\\
44	85.76\\
45	85.7166666666667\\
46	85.5866666666667\\
47	85.6133333333333\\
48	85.6833333333334\\
49	85.7666666666667\\
50	85.78\\
};
\addlegendentry{Lasso-based}

\addplot [color=mycolor1, dashed, line width=1.0pt]
table[row sep=crcr]{%
	1	67.06\\
	2	74.2366666666667\\
	3	75.81\\
	4	76.8833333333333\\
	5	77.8666666666667\\
	6	78.88\\
	7	79.6933333333333\\
	8	80.4066666666667\\
	9	81.0833333333333\\
	10	81.6533333333333\\
	11	82.1133333333333\\
	12	82.3733333333333\\
	13	82.6166666666667\\
	14	82.9933333333333\\
	15	83.2266666666667\\
	16	83.5166666666667\\
	17	83.73\\
	18	83.9766666666667\\
	19	84.2333333333333\\
	20	84.44\\
	21	84.69\\
	22	84.78\\
	23	84.9533333333333\\
	24	85.17\\
	25	85.26\\
	26	85.4233333333333\\
	27	85.5866666666667\\
	28	85.69\\
	29	85.79\\
	30	85.94\\
	31	85.93\\
	32	86.0866666666667\\
	33	86.0833333333333\\
	34	86.24\\
	35	86.3466666666667\\
	36	86.3366666666667\\
	37	86.3933333333333\\
	38	86.41\\
	39	86.58\\
	40	86.62\\
	41	86.6966666666667\\
	42	86.7133333333333\\
	43	86.8266666666667\\
	44	86.97\\
	45	87.0433333333333\\
	46	87.0366666666667\\
	47	87.1133333333333\\
	48	87.1733333333333\\
	49	87.24\\
	50	87.3533333333333\\
};
\addlegendentry{All Antenna selection}
\end{axis}

\end{tikzpicture}%

%% file: fmnistn16.tex
%
%
\definecolor{mycolor1}{rgb}{0.00000,0.49804,0.00000}%
\definecolor{mycolor2}{rgb}{1.00000,0.00000,1.00000}%
\definecolor{mycolor3}{rgb}{0.00000,0.44706,0.74118}%
\begin{tikzpicture}

\begin{axis}[%
width=2.4in,
height=3.2in,
at={(2.6in,1.024in)},
scale only axis,
xmin=1,
xmax=50,
xlabel style={font=\color{white!15!black}},
xlabel={Communication rounds},
ymin=55,
ymax=90,
ylabel style={font=\color{white!15!black}},
ylabel={Testing accuracy in $\%$},
xlabel style={ font=\small},
ylabel style={ font=\small},
axis background/.style={fill=white},
ticklabel style={font=\small},
xtick={0,10,20,30,40,50},
xticklabels={{$0$},{$10$},{$20$},{$30$},{$40$},{$50$}},
ytick={55,60,65,70,75,80,85,90},
yticklabels={{$55$},{$60$},{$65$},{$70$},{$75$},{$80$},{$85$},{$90$}},
xmajorgrids,
ymajorgrids,
legend style={at={(0.97,0.03)}, nodes={scale=0.85, transform shape}, anchor=south east, legend cell align=left, align=left, draw=white!15!black}
]
\addplot [color=black, line width=1.0pt, mark=asterisk, mark options={solid, black}]
  table[row sep=crcr]{%
1	29.935\\
2	49.54\\
3	53.95\\
4	60.605\\
5	63.95\\
6	66.32\\
7	68.39\\
8	69.45\\
9	70.31\\
10	70.805\\
11	72.155\\
12	72.485\\
13	72.69\\
14	73.05\\
15	73.54\\
16	73.93\\
17	73.995\\
18	73.995\\
19	73.945\\
20	74.09\\
21	74.92\\
22	75.4\\
23	75.445\\
24	75.77\\
25	75.71\\
26	76.115\\
27	76.365\\
28	76.185\\
29	76.26\\
30	76.38\\
31	76.51\\
32	76.425\\
33	76.205\\
34	76.235\\
35	76.3\\
36	76.27\\
37	76.045\\
38	76.175\\
39	76.13\\
40	76.165\\
41	76.405\\
42	76.765\\
43	76.71\\
44	76.81\\
45	76.825\\
46	76.84\\
47	76.885\\
48	76.835\\
49	77.005\\
50	77.085\\
};
\addlegendentry{Randoom Policy}

\addplot [color=mycolor1, line width=1.0pt, mark=triangle, mark options={solid, rotate=180, mycolor1}]
  table[row sep=crcr]{%
1	36.3\\
2	59.05\\
3	55.07\\
4	59.71\\
5	64.41\\
6	65.65\\
7	68.74\\
8	70.65\\
9	71.93\\
10	72.99\\
11	74.29\\
12	75.16\\
13	75.28\\
14	75.74\\
15	76.52\\
16	76.61\\
17	76.8\\
18	76.63\\
19	76.67\\
20	77\\
21	77.88\\
22	78.12\\
23	77.5\\
24	77.84\\
25	77.43\\
26	77.98\\
27	78.35\\
28	77.75\\
29	77.79\\
30	77.6\\
31	77.88\\
32	78.14\\
33	78.05\\
34	78.16\\
35	78.39\\
36	78.33\\
37	78.26\\
38	78.45\\
39	78.48\\
40	78.44\\
41	78.67\\
42	78.715\\
43	78.741\\
44	78.872\\
45	78.895\\
46	78.919\\
47	78.935\\
48	78.936\\
49	78.969\\
50	78.96\\
};
\addlegendentry{Greedy Antenna Selection}

\addplot [color=red, line width=1.0pt, mark=o, mark options={solid, red}]
table[row sep=crcr]{%
	1	46.4\\
	2	62.22\\
	3	65.67\\
	4	71.34\\
	5	72.55\\
	6	73.53\\
	7	74.79\\
	8	74.96\\
	9	75.82\\
	10	76.68\\
	11	77.51\\
	12	77.99\\
	13	78.4\\
	14	78.62\\
	15	78.64\\
	16	78.7\\
	17	78.75\\
	18	78.89\\
	19	78.9\\
	20	78.95\\
	21	78.99\\
	22	79\\
	23	79.188\\
	24	79.2096\\
	25	79.2814\\
	26	79.3863\\
	27	79.4088\\
	28	79.5862\\
	29	79.84\\
	30	79.9\\
	31	79.9\\
	32	80.0946\\
	33	80.1962\\
	34	80.209\\
	35	80.259\\
	36	80.308\\
	37	80.40862\\
	38	80.43\\
	39	80.448\\
	40	80.8\\
	41	80.7\\
	42	80.9\\
	43	80.962\\
	44	80.87\\
	45	80.945\\
	46	80.837\\
	47	80.94\\
	48	80.943\\
	49	81\\
	50	81\\
};
\addlegendentry{PDD-based }

\addplot [color=blue, line width=1.0pt, mark=x, mark options={solid, blue}]
  table[row sep=crcr]{%
1	46.4\\
2	67.36\\
3	70.47\\
4	72.79\\
5	73.93\\
6	75.38\\
7	76.22\\
8	76.76\\
9	77.08\\
10	77.41\\
11	77.73\\
12	78.34\\
13	78.47\\
14	78.86\\
15	79.16\\
16	79.38\\
17	79.46\\
18	79.74\\
19	79.94\\
20	80.03\\
21	80.35\\
22	80.63\\
23	80.97\\
24	81.2\\
25	81.37\\
26	81.48\\
27	81.71\\
28	81.9\\
29	82.02\\
30	82.14\\
31	81.84\\
32	81.94\\
33	81.97\\
34	81.73\\
35	81.94\\
36	82\\
37	81.98\\
38	82.14\\
39	82.16\\
40	82.07\\
41	82.28\\
42	82.07\\
43	82.24\\
44	82.29\\
45	82.21\\
46	82.31\\
47	82.294\\
48	82.285\\
49	82.24\\
50	82.29\\
};
\addlegendentry{FISTA}

\addplot [color=mycolor2, line width=1.0pt, mark=star, mark options={solid, mycolor2}]
  table[row sep=crcr]{%
1	51.45\\
2	71.72\\
3	73.34\\
4	74.47\\
5	75.54\\
6	76.25\\
7	77.27\\
8	77.92\\
9	78.42\\
10	78.88\\
11	79.11\\
12	79.08\\
13	79.87\\
14	80.16\\
15	80.51\\
16	80.49\\
17	80.9\\
18	80.96\\
19	81.2\\
20	81.43\\
21	81.75\\
22	81.74\\
23	82.02\\
24	82.16\\
25	82.15\\
26	82.19\\
27	82.32\\
28	82.27\\
29	82.41\\
30	82.56\\
31	82.63\\
32	82.65\\
33	82.8\\
34	82.84\\
35	82.79\\
36	82.82\\
37	82.94\\
38	83.18\\
39	83.02\\
40	83.05\\
41	83.2\\
42	83.21\\
43	83\\
44	83.02\\
45	83.31\\
46	83.14\\
47	83.37\\
48	83.32\\
49	83.19\\
50	83.31\\
};
\addlegendentry{Lasso-based}

\addplot [color=mycolor3, dashed, line width=1.0pt]
  table[row sep=crcr]{%
1	62.6233333333333\\
2	71.6366666666667\\
3	73.57\\
4	74.9166666666667\\
5	76.0466666666667\\
6	77.0433333333333\\
7	77.9933333333333\\
8	78.7533333333333\\
9	79.61\\
10	80.1\\
11	80.6233333333333\\
12	81.0533333333333\\
13	81.5666666666667\\
14	81.8\\
15	82.2133333333333\\
16	82.2666666666667\\
17	82.43\\
18	82.5766666666667\\
19	82.75\\
20	82.9433333333333\\
21	83\\
22	83.19\\
23	83.23\\
24	83.3866666666667\\
25	83.43\\
26	83.64\\
27	83.6533333333333\\
28	83.78\\
29	83.8633333333333\\
30	84.02\\
31	84.0733333333333\\
32	84.12\\
33	84.0933333333333\\
34	84.1733333333333\\
35	84.1633333333333\\
36	84.41\\
37	84.4133333333333\\
38	84.5066666666666\\
39	84.4366666666667\\
40	84.6066666666667\\
41	84.5566666666667\\
42	84.6733333333333\\
43	84.75\\
44	84.8966666666667\\
45	84.8133333333333\\
46	84.9666666666667\\
47	84.9133333333333\\
48	85.0966666666667\\
49	85.1366666666667\\
50	85.01\\
};
\addlegendentry{All Antenna selection}

\end{axis}
\end{tikzpicture}%

%% file: fmnistn40.tex
%
%
\definecolor{mycolor1}{rgb}{0.00000,0.44700,0.74100}%
\definecolor{mycolor2}{rgb}{1.00000,0.00000,1.00000}%
\definecolor{mycolor3}{rgb}{0.00000,0.49804,0.00000}%
\begin{tikzpicture}

\begin{axis}[%
width=2.4in,
height=3.2in,
at={(2.6in,1.024in)},
scale only axis,
xmin=0.9,
xmax=50,
xlabel style={font=\color{white!15!black}},
xlabel={Communication rounds},
ymin=55,
ymax=90,
ylabel style={font=\color{white!15!black}},
ylabel={Testing accuracy in $\%$},
xlabel style={ font=\small},
ylabel style={ font=\small},
ticklabel style={font=\small},
xtick={0,10,20,30,40,50},
xticklabels={{$0$},{$10$},{$20$},{$30$},{$40$},{$50$}},
ytick={55,60,65,70,75,80,85,90},
yticklabels={{$55$},{$60$},{$65$},{$70$},{$75$},{$80$},{$85$},{$90$}},
axis background/.style={fill=white},
xmajorgrids,
ymajorgrids,
legend style={at={(0.97,0.03)}, anchor=south east,  nodes={scale=0.85, transform shape}, legend cell align=left, align=left, draw=white!15!black}
]

\addplot [color=black, line width=1.0pt, mark=asterisk, mark options={solid, black}]
  table[row sep=crcr]{%
1	61.18\\
2	59.25\\
3	56.99\\
4	66.02\\
5	71.51\\
6	73.4\\
7	74.36\\
8	74.78\\
9	75.35\\
10	75.94\\
11	76.37\\
12	77.31\\
13	77.58\\
14	77.63\\
15	77.94\\
16	78.55\\
17	78.66\\
18	78.77\\
19	78.89\\
20	78.88\\
21	78.95\\
22	79.16\\
23	79\\
24	79.15\\
25	79.06\\
26	79.42\\
27	79.11\\
28	79.2\\
29	79.19\\
30	79.19\\
31	79.27\\
32	79.7\\
33	79.78\\
34	80.07\\
35	79.78\\
36	80.11\\
37	80.11\\
38	80.21\\
39	80.07\\
40	80.17\\
41	79.75\\
42	79.91\\
43	79.84\\
44	79.7\\
45	79.59\\
46	79.61\\
47	79.44\\
48	79.56\\
49	79.62\\
50	79.39\\
};
\addlegendentry{Random Policy}

\addplot [color=mycolor3, line width=1.0pt, mark=triangle, mark options={solid, rotate=180, mycolor3}]
  table[row sep=crcr]{%
1	59.25\\
2	69.21\\
3	71.37\\
4	72.36\\
5	73.44\\
6	74.32\\
7	75.12\\
8	75.83\\
9	76.19\\
10	77.04\\
11	77.58\\
12	77.92\\
13	78.07\\
14	78.37\\
15	78.76\\
16	78.88\\
17	79.07\\
18	79.5\\
19	79.62\\
20	79.7\\
21	79.85\\
22	79.83\\
23	79.58\\
24	79.83\\
25	79.85\\
26	79.66\\
27	79.67\\
28	79.97\\
29	80.07\\
30	80.09\\
31	80.07\\
32	80.3\\
33	80.51\\
34	80.67\\
35	80.84\\
36	81.04\\
37	80.81\\
38	80.76\\
39	80.86\\
40	81.04\\
41	81.05\\
42	81.04\\
43	81.1\\
44	81.15\\
45	81.14\\
46	80.83\\
47	80.99\\
48	81.14\\
49	81.27\\
50	81.15\\
};
\addlegendentry{Greedy Antenna Selection}

\addplot [color=red, line width=1.0pt, mark=o, mark options={solid, red}]
table[row sep=crcr]{%
	1	60.6766666666667\\
	2	70.9533333333333\\
	3	72.5\\
	4	73.9166666666667\\
	5	74.89\\
	6	75.7233333333333\\
	7	76.4333333333333\\
	8	77.5133333333333\\
	9	78.3233333333333\\
	10	78.98\\
	11	79.57\\
	12	80.0866666666667\\
	13	80.4233333333333\\
	14	80.83\\
	15	81.1666666666667\\
	16	81.4766666666667\\
	17	81.7166666666667\\
	18	81.8733333333333\\
	19	81.9533333333333\\
	20	82.1533333333333\\
	21	82.38\\
	22	82.4966666666667\\
	23	82.6333333333333\\
	24	82.7333333333333\\
	25	82.8433333333333\\
	26	83.02\\
	27	83.1333333333333\\
	28	83.2766666666667\\
	29	83.4566666666667\\
	30	83.42\\
	31	83.5366666666667\\
	32	83.4633333333333\\
	33	83.52\\
	34	83.5466666666667\\
	35	83.5133333333333\\
	36	83.74\\
	37	83.8033333333333\\
	38	83.8066666666667\\
	39	83.71\\
	40	83.7966666666667\\
	41	83.94\\
	42	83.95\\
	43	84.0433333333333\\
	44	84.1033333333333\\
	45	84.0966666666666\\
	46	84.1566666666667\\
	47	84.1533333333333\\
	48	84.17\\
	49	84.2166666666667\\
	50	84.1466666666667\\
};
\addlegendentry{PDD-based }

\addplot [color=blue, line width=1.0pt, mark=x, mark options={solid, blue}]
table[row sep=crcr]{%
	1	62.2366666666667\\
	2	69.91\\
	3	71.4833333333333\\
	4	73.1833333333333\\
	5	74.21\\
	6	75.17\\
	7	75.9666666666667\\
	8	76.6766666666667\\
	9	77.2433333333333\\
	10	77.8766666666667\\
	11	78.43\\
	12	78.9433333333333\\
	13	79.42\\
	14	79.7733333333333\\
	15	80.2266666666667\\
	16	80.3633333333333\\
	17	80.53\\
	18	80.9233333333333\\
	19	81.1166666666667\\
	20	81.2966666666667\\
	21	81.4533333333333\\
	22	81.5766666666667\\
	23	81.8133333333333\\
	24	81.83\\
	25	82.02\\
	26	82.04\\
	27	82.17\\
	28	82.16\\
	29	82.1966666666667\\
	30	82.3866666666667\\
	31	82.4233333333333\\
	32	82.2866666666667\\
	33	82.3966666666667\\
	34	82.4466666666667\\
	35	82.4333333333333\\
	36	82.4\\
	37	82.4866666666667\\
	38	82.4966666666667\\
	39	82.47\\
	40	82.64\\
	41	82.6433333333333\\
	42	82.5833333333333\\
	43	82.6033333333334\\
	44	82.5766666666667\\
	45	82.6866666666667\\
	46	82.7033333333333\\
	47	82.66\\
	48	82.6433333333333\\
	49	82.6733333333333\\
	50	82.62\\
};
\addlegendentry{FISTA}

\addplot [color=mycolor2, line width=1.0pt, mark=star, mark options={solid, mycolor2}]
table[row sep=crcr]{%
	1	61.38\\
	2	70.03\\
	3	71.02\\
	4	73.29\\
	5	74.4\\
	6	75.49\\
	7	76.2\\
	8	77.05\\
	9	77.51\\
	10	78.21\\
	11	78.8\\
	12	79.31\\
	13	79.82\\
	14	80.31\\
	15	80.64\\
	16	80.9\\
	17	81.04\\
	18	81.02\\
	19	81.41\\
	20	81.68\\
	21	81.83\\
	22	81.99\\
	23	82.17\\
	24	82.15\\
	25	82.34\\
	26	82.59\\
	27	82.54\\
	28	82.51\\
	29	82.64\\
	30	82.73\\
	31	82.81\\
	32	82.57\\
	33	82.76\\
	34	82.7\\
	35	82.91\\
	36	82.66\\
	37	82.85\\
	38	82.76\\
	39	82.66\\
	40	82.83\\
	41	82.78\\
	42	82.7\\
	43	82.76\\
	44	82.89\\
	45	83.07\\
	46	83.19\\
	47	83.17\\
	48	83.05\\
	49	83.02\\
	50	83.06\\
};
\addlegendentry{Lasso based }

\addplot [color=mycolor1, dashed, line width=1.0pt]
table[row sep=crcr]{%
	1	62.6233333333333\\
	2	71.6366666666667\\
	3	73.57\\
	4	74.9166666666667\\
	5	76.0466666666667\\
	6	77.0433333333333\\
	7	77.9933333333333\\
	8	78.7533333333333\\
	9	79.61\\
	10	80.1\\
	11	80.6233333333333\\
	12	81.0533333333333\\
	13	81.5666666666667\\
	14	81.8\\
	15	82.2133333333333\\
	16	82.2666666666667\\
	17	82.43\\
	18	82.5766666666667\\
	19	82.75\\
	20	82.9433333333333\\
	21	83\\
	22	83.19\\
	23	83.23\\
	24	83.3866666666667\\
	25	83.43\\
	26	83.64\\
	27	83.6533333333333\\
	28	83.78\\
	29	83.8633333333333\\
	30	84.02\\
	31	84.0733333333333\\
	32	84.12\\
	33	84.0933333333333\\
	34	84.1733333333333\\
	35	84.1633333333333\\
	36	84.41\\
	37	84.4133333333333\\
	38	84.5066666666666\\
	39	84.4366666666667\\
	40	84.6066666666667\\
	41	84.5566666666667\\
	42	84.6733333333333\\
	43	84.75\\
	44	84.8966666666667\\
	45	84.8133333333333\\
	46	84.9666666666667\\
	47	84.9133333333333\\
	48	85.0966666666667\\
	49	85.1366666666667\\
	50	85.01\\
};
\addlegendentry{All Antenna selection}
\end{axis}

\end{tikzpicture}%

%% file: cifarL64iid.tex
%
%
\definecolor{mycolor1}{rgb}{0.00000,0.44706,0.74118}%
\definecolor{mycolor2}{rgb}{0.00000,0.49804,0.00000}%
\definecolor{mycolor3}{rgb}{1.00000,0.00000,1.00000}%
\begin{tikzpicture}

\begin{axis}[%
width=2.4in,
height=3.2in,
at={(2.6in,1.025in)},
scale only axis,
xlabel style={font=\color{white!15!black}},
xmin=0.1,
xmax=50.1,
ymin=9.86,
ymax=70.08271565109572,
ticklabel style={font=\small},
xtick={0,10,20,30,40,50},
xticklabels={{$0$},{$10$},{$20$},{$30$},{$40$},{$50$}},
ytick={0,10,20,30,40,50,60,70},
yticklabels={{$0$},{$10$},{$20$},{$30$},{$40$},{$50$},{$60$},{$70$}},
xlabel={Communication rounds},
ylabel={Testing accuracy in $\%$},
xlabel style={ font=\small},
ylabel style={ font=\small},
axis background/.style={fill=white},
xmajorgrids,
ymajorgrids,
legend style={at={(0.97,0.03)},  nodes={scale=0.85, transform shape}, anchor=south east, legend cell align=left, align=left, draw=white!15!black}
]

\addplot [color=black, line width=1.0pt, mark=asterisk, mark options={solid, black}]
  table[row sep=crcr]{%
1	9.86\\
2	13\\
3	25.14\\
4	27.85\\
5	30.22\\
6	33\\
7	36.07\\
8	38.38\\
9	39.66\\
10	40.84\\
11	41.93\\
12	43.29\\
13	44.36\\
14	45.03\\
15	45.51\\
16	46.14\\
17	46.61\\
18	46.9\\
19	47.55\\
20	47.81\\
21	47.81\\
22	47.72\\
23	48.37\\
24	48.63\\
25	49.26\\
26	49.73\\
27	49.65\\
28	49.87\\
29	49.92\\
30	50.19\\
31	50.21\\
32	50.45\\
33	50.87\\
34	51.09\\
35	51.64\\
36	52.01\\
37	52.74\\
38	53.17\\
39	53.6\\
40	53.85\\
41	54.14\\
42	53.96\\
43	54.44\\
44	54.35\\
45	54.52\\
46	54.69\\
47	54.81\\
48	54.98\\
49	55.21\\
50	55.53\\
};
\addlegendentry{Random Policy}

\addplot [color=mycolor2, line width=1.0pt, mark=triangle, mark options={solid, rotate=180, mycolor2}]
  table[row sep=crcr]{%
1	10.64\\
2	15.0675\\
3	25.67\\
4	28.47\\
5	31.01\\
6	34.1925\\
7	36.795\\
8	38.69\\
9	39.9325\\
10	40.9925\\
11	42.12\\
12	43.3125\\
13	44.225\\
14	45.38\\
15	46.175\\
16	46.9075\\
17	47.6275\\
18	48.2225\\
19	48.595\\
20	49.1475\\
21	49.705\\
22	50.2475\\
23	50.56\\
24	51.02\\
25	51.35\\
26	51.6625\\
27	51.9625\\
28	52.41\\
29	52.9625\\
30	53.175\\
31	53.4075\\
32	53.74\\
33	53.975\\
34	54.21\\
35	54.3625\\
36	54.5875\\
37	54.7525\\
38	54.85\\
39	55.26\\
40	55.6475\\
41	55.7875\\
42	56.055\\
43	56.3975\\
44	56.6\\
45	56.6975\\
46	56.805\\
47	56.9725\\
48	57.085\\
49	57.2075\\
50	57.2675\\
};
\addlegendentry{Greedy Antenna Selection}

\addplot [color=red, line width=1.0pt, mark=o, mark options={solid, red}]
table[row sep=crcr]{%
	1	14.715\\
	2	29.145\\
	3	34.845\\
	4	37.32\\
	5	38.775\\
	6	40.69\\
	7	42.63\\
	8	43.71\\
	9	44.01\\
	10	45.37\\
	11	45.85\\
	12	47.835\\
	13	48.555\\
	14	50.61\\
	15	52.22\\
	16	53.26\\
	17	53.96\\
	18	54.885\\
	19	55.795\\
	20	56.53\\
	21	57.025\\
	22	57.51\\
	23	58.055\\
	24	58.445\\
	25	58.815\\
	26	59.295\\
	27	59.535\\
	28	59.7\\
	29	60.25\\
	30	60.805\\
	31	60.555\\
	32	60.935\\
	33	61.305\\
	34	61.78\\
	35	61.425\\
	36	61.655\\
	37	61.805\\
	38	61.915\\
	39	62.09\\
	40	62.17\\
	41	62.205\\
	42	62.39\\
	43	62.605\\
	44	62.825\\
	45	62.815\\
	46	63.595\\
	47	63.955\\
	48	64.015\\
	49	64.265\\
	50	64.175\\
};
\addlegendentry{PDD-based}
\addplot [color=blue, line width=1.0pt, mark=x, mark options={solid, blue}]
table[row sep=crcr]{%
	1	14.84\\
	2	18.74\\
	3	29.7975\\
	4	32.2625\\
	5	33.765\\
	6	36.655\\
	7	39.835\\
	8	41.93\\
	9	43.1675\\
	10	44.095\\
	11	45.335\\
	12	46.5075\\
	13	47.7125\\
	14	48.635\\
	15	49.24\\
	16	49.9325\\
	17	50.3325\\
	18	50.9925\\
	19	51.7\\
	20	52.4025\\
	21	52.7925\\
	22	53.44\\
	23	53.8975\\
	24	54.285\\
	25	54.7475\\
	26	55.365\\
	27	55.3075\\
	28	55.7425\\
	29	56.12\\
	30	56.9325\\
	31	57.3475\\
	32	58.005\\
	33	58.405\\
	34	58.58\\
	35	58.9575\\
	36	59.425\\
	37	59.605\\
	38	59.82\\
	39	60.3175\\
	40	60.67\\
	41	61.075\\
	42	60.9975\\
	43	61.22\\
	44	61.19775\\
	45	61.18275\\
	46	61.278\\
	47	61.28575\\
	48	61.26\\
	49	61.375\\
	50	61.5325\\
};
\addlegendentry{FISTA}

\addplot [color=mycolor3, line width=1.0pt, mark=star, mark options={solid, mycolor3}]
  table[row sep=crcr]{%
1	15.52\\
2	19.32\\
3	30.84\\
4	33.7\\
5	36.04\\
6	39.77\\
7	42.59\\
8	44.56\\
9	45.57\\
10	46.67\\
11	48.05\\
12	49.27\\
13	50.03\\
14	50.94\\
15	51.33\\
16	51.93\\
17	52.32\\
18	52.89\\
19	53.57\\
20	53.72\\
21	54.06\\
22	55.07\\
23	55.41\\
24	55.84\\
25	56.52\\
26	56.7\\
27	57.25\\
28	57.84\\
29	58.35\\
30	58.45\\
31	58.62\\
32	58.86\\
33	59.25\\
34	59.49\\
35	59.65\\
36	59.69\\
37	59.96\\
38	60.23\\
39	60.85\\
40	61.08\\
41	60.92\\
42	61.11\\
43	61.58\\
44	61.73\\
45	61.97\\
46	61.95\\
47	61.88\\
48	62.05\\
49	62.06\\
50	62.34\\
};
\addlegendentry{Lasso-based }

\addplot [color=mycolor1, dashed, line width=1.0pt]
table[row sep=crcr]{%
	1	11.15\\
	2	29.34\\
	3	35.17\\
	4	38.53\\
	5	40.44\\
	6	42.87\\
	7	44.86\\
	8	46.53\\
	9	48.22\\
	10	48.96\\
	11	49.24\\
	12	49.69\\
	13	51.74\\
	14	53.07\\
	15	53.78\\
	16	54.64\\
	17	56.41\\
	18	56.87\\
	19	57.26\\
	20	57.5\\
	21	58.16\\
	22	58.18\\
	23	58.97\\
	24	59.3\\
	25	60.09\\
	26	60.11\\
	27	60.53\\
	28	60.69\\
	29	60.96\\
	30	61.78\\
	31	62.04\\
	32	62.33\\
	33	62.74\\
	34	62.91\\
	35	63.28\\
	36	63.31\\
	37	63.72\\
	38	64.39\\
	39	64.05\\
	40	64.28\\
	41	64.29\\
	42	64.24\\
	43	64.61\\
	44	64.92\\
	45	64.99\\
	46	65.02\\
	47	65.3\\
	48	65.27\\
	49	65.44\\
	50	65.33\\
};
\addlegendentry{All Antenna selection}
\end{axis}
\end{tikzpicture}%

%% file: cifar64.tex
%
%
\definecolor{mycolor1}{rgb}{0.00000,0.44706,0.74118}%
\definecolor{mycolor2}{rgb}{0.00000,0.49804,0.00000}%
\definecolor{mycolor3}{rgb}{1.00000,0.00000,1.00000}%
\begin{tikzpicture}

\begin{axis}[%
width=2.4in,
height=3.2in,
at={(2.6in,1.024in)},
scale only axis,
xmin=1,
xmax=50,
ymin=9,
ymax=60,
xlabel={Communication rounds},
ticklabel style={font=\small},
xtick={0,10,20,30,40,50},
xticklabels={{$0$},{$10$},{$20$},{$30$},{$40$},{$50$}},
ytick={0,10,20,30,40,50,60},
yticklabels={{$0$},{$10$},{$20$},{$30$},{$40$},{$50$},{$60$}},
ylabel={Testing accuracy in $\%$},
xlabel style={ font=\small},
ylabel style={ font=\small},
axis background/.style={fill=white},
xmajorgrids,
ymajorgrids,
legend style={at={(0.18,0.03)},  nodes={scale=0.85, transform shape}, anchor=south west, legend cell align=left, align=left, draw=white!15!black}
]

\addplot [color=black, line width=1.0pt, mark=asterisk, mark options={solid, black}]
table[row sep=crcr]{%
	1	15.84\\
	2	18.32\\
	3	18.6025\\
	4	20.3575\\
	5	20.2925\\
	6	20.8325\\
	7	22.8325\\
	8	23.415\\
	9	24.335\\
	10	24.0125\\
	11	25.335\\
	12	26.8\\
	13	27.04\\
	14	28.0875\\
	15	28.7575\\
	16	29.4975\\
	17	29.61\\
	18	29.89\\
	19	29.58\\
	20	29.305\\
	21	30.155\\
	22	32.725\\
	23	33.3075\\
	24	33.7475\\
	25	34.0525\\
	26	35.245\\
	27	36.44\\
	28	36.62725\\
	29	36.7775\\
	30	36.925\\
	31	37.325\\
	32	37.5275\\
	33	36.96\\
	34	36.9775\\
	35	37.7825\\
	36	38.865\\
	37	38.665\\
	38	38.1075\\
	39	38.2675\\
	40	37.9175\\
	41	37.435\\
	42	38.145\\
	43	39.045\\
	44	39.2225\\
	45	38.77\\
	46	39.1475\\
	47	39.0925\\
	48	39.52\\
	49	39.7325\\
	50	39.215\\
};
\addlegendentry{Random Policy}
\addplot [color=mycolor2, line width=1.0pt, mark=triangle, mark options={solid, rotate=180, mycolor2}]
table[row sep=crcr]{%
	1	14.8\\
	2	24.86\\
	3	26.06\\
	4	25.73\\
	5	26.61\\
	6	27.98\\
	7	28.52\\
	8	29.36\\
	9	30.61\\
	10	31.19\\
	11	34.502\\
	12	35.31\\
	13	34.75\\
	14	34.91\\
	15	35.87\\
	16	36.64\\
	17	37.7\\
	18	37.9\\
	19	38.1\\
	20	38.33\\
	21	39.12\\
	22	40.79\\
	23	40.7\\
	24	40.68\\
	25	41.15\\
	26	41.42\\
	27	42.12\\
	28	41.69\\
	29	42.21\\
	30	42.72\\
	31	42.74\\
	32	42.91\\
	33	43.21\\
	34	44.08\\
	35	45.87\\
	36	45.42\\
	37	45.5\\
	38	45.88\\
	39	45.88\\
	40	46\\
	41	46.1\\
	42	46.5\\
	43	45\\
	44	46.6\\
	45	46.7\\
	46	47\\
	47	46\\
	48	46\\
	49	46.5\\
	50	47.01\\
};
\addlegendentry{Greedy Antenna Selection}

\addplot [color=red, line width=1.0pt, mark=o, mark options={solid, red}]
table[row sep=crcr]{%
	1	7.74\\
	2	11.63\\
	3	23.56\\
	4	25.8\\
	5	28.24\\
	6	30.95\\
	7	34.47\\
	8	36.62\\
	9	38.19\\
	10	38.78\\
	11	39.68\\
	12	40.93\\
	13	42.04\\
	14	43.88\\
	15	44.37\\
	16	44.88\\
	17	45.78\\
	18	46.68\\
	19	47.3\\
	20	47.84\\
	21	48.56\\
	22	49.04\\
	23	49.31\\
	24	50\\
	25	50.5\\
	26	50.36\\
	27	50.53\\
	28	51.11\\
	29	51\\
	30	51.19\\
	31	51.93\\
	32	51.92\\
	33	52.42\\
	34	52.66\\
	35	52.47\\
	36	53.19\\
	37	53.66\\
	38	53.3\\
	39	53.3\\
	40	53.45\\
	41	53.83\\
	42	53.93\\
	43	53.92\\
	44	54.07\\
	45	54.03\\
	46	54\\
	47	54.26\\
	48	54.28\\
	49	54.48\\
	50	54.62\\
};
\addlegendentry{PDD-based}

\addplot [color=blue, line width=1.0pt, mark=x, mark options={solid, blue}]
  table[row sep=crcr]{%
1	5.84\\
2	9.74\\
3	20.7975\\
4	23.2625\\
5	24.765\\
6	27.655\\
7	30.835\\
8	32.93\\
9	34.1675\\
10	35.095\\
11	36.335\\
12	37.5075\\
13	38.7125\\
14	39.635\\
15	40.24\\
16	40.9325\\
17	41.3325\\
18	41.9925\\
19	42.7\\
20	43.4025\\
21	43.7925\\
22	44.44\\
23	44.8975\\
24	45.285\\
25	45.7475\\
26	46.365\\
27	46.3075\\
28	46.7425\\
29	47.12\\
30	47.9325\\
31	48.3475\\
32	49.005\\
33	49.405\\
34	49.58\\
35	49.9575\\
36	50.425\\
37	50.605\\
38	50.82\\
39	51.3175\\
40	51.67\\
41	52.075\\
42	51.9975\\
43	52.22\\
44	51.9775\\
45	51.8275\\
46	51.78\\
47	51.8575\\
48	52.06\\
49	52.175\\
50	52.5325\\
};
\addlegendentry{FISTA}

\addplot [color=mycolor3, line width=1.0pt, mark=star, mark options={solid, mycolor3}]
table[row sep=crcr]{%
1	5.73\\
2	8.95\\
3	20.56\\
4	23.39\\
5	24.93\\
6	27.77\\
7	30.99\\
8	33.85\\
9	35.12\\
10	35.84\\
11	36.64\\
12	38.16\\
13	39.27\\
14	40.62\\
15	41.14\\
16	41.59\\
17	42.19\\
18	42.65\\
19	44.02\\
20	44.6\\
21	45.2\\
22	45.71\\
23	46.46\\
24	46.43\\
25	47.16\\
26	47.88\\
27	47.87\\
28	48.46\\
29	48.74\\
30	49.49\\
31	50.34\\
32	50.53\\
33	50.84\\
34	50.9\\
35	51\\
36	51.68\\
37	51.42\\
38	51.37\\
39	51.53\\
40	51.92\\
41	52.72\\
42	52.65\\
43	52.63\\
44	52.34\\
45	52.18\\
46	52.36\\
47	52.74\\
48	52.26\\
49	52.6\\
50	52.68\\
};
\addlegendentry{Lasso-based}

\addplot [color=mycolor1, dashed, line width=1.0pt]
table[row sep=crcr]{%
	1	10.64\\
	2	15.0675\\
	3	25.67\\
	4	28.47\\
	5	31.01\\
	6	34.1925\\
	7	36.795\\
	8	38.69\\
	9	39.9325\\
	10	40.9925\\
	11	42.12\\
	12	43.3125\\
	13	44.225\\
	14	45.38\\
	15	46.175\\
	16	46.9075\\
	17	47.6275\\
	18	48.2225\\
	19	48.595\\
	20	49.1475\\
	21	49.705\\
	22	50.2475\\
	23	50.56\\
	24	51.02\\
	25	51.35\\
	26	51.6625\\
	27	51.9625\\
	28	52.41\\
	29	52.9625\\
	30	53.175\\
	31	53.4075\\
	32	53.74\\
	33	53.975\\
	34	54.21\\
	35	54.3625\\
	36	54.5875\\
	37	54.7525\\
	38	54.85\\
	39	55.26\\
	40	55.6475\\
	41	55.7875\\
	42	56.055\\
	43	56.3975\\
	44	56.6\\
	45	56.6975\\
	46	56.805\\
	47	56.9725\\
	48	57.085\\
	49	57.2075\\
	50	57.2675\\
};
\addlegendentry{All Antenna selection}
\end{axis}
\end{tikzpicture}%